\documentclass[a4paper,UKenglish,cleveref, autoref, thm-restate]{lipics-v2021}

\pdfoutput=1 
\hideLIPIcs  
\nolinenumbers

\usepackage{latexsym}
\usepackage[fleqn]{amsmath}
\usepackage{newtxtext}
\usepackage{thmtools}
\usepackage{fixme}
\usepackage{amssymb}
\usepackage{mathrsfs}
\usepackage{amssymb}
\usepackage{amsthm}
\usepackage{mathtools}

\newcommand{\poly}{\mathrm{poly}}
\newcommand{\defeq}{\stackrel{\mbox{\scriptsize{\normalfont\rmfamily def}}}{=}}

\usepackage{bm}
\renewcommand{\Vec}[1]{\bm{#1}}

\newcommand{\np}{\mathrm{NP}}
\newcommand{\xp}{\mathrm{XP}}
\newcommand{\sharpp}{\#\mathrm{P}}
\newcommand{\distfamily}{\mathcal{LU}}
\newcommand{\pfs}{\mathcal{F}}
\newcommand{\graphzero}{G}
\newcommand{\graphone}{G^+}
\newcommand{\graphtwo}{\tilde{G}}
\newcommand{\graphthree}{G^\ast}
\newcommand{\graphfour}{\acute{G}}
\newcommand{\graphfive}{\hat G}
\newcommand{\graphsix}{\check G}

\newcommand{\Unif}{\mathrm{Unif}}
\newcommand{\unif}{\mathrm{unif}}
\newcommand{\iidsim}{\overset{\mathrm{iid}}{\sim}}
\newcommand{\cat}[2]{\left[ #1 ; #2 \right]}

\title{The Complexity of Computing Path Length Distributions with Edges i.i.d. Random via Local Uniformity} 

\titlerunning{Path Length Distributions with Locally Uniform Edge Lengths} 

\author{Ei ANDO }{Senshu University, 2-1-1, Higashimita, Tama-Ku, Kawasaki, Kanagawa, Japan}{ando.ei@isc.senshu-u.ac.jp}{https://orcid.org/0000-0001-7284-9896}{
}

\authorrunning{E. ANDO} 

\Copyright{Ei ANDO} 

\ccsdesc[500]{Theory of computation~Shortest paths}
\ccsdesc[500]{Theory of computation~Problems, reductions and completeness}
\ccsdesc[500]{Mathematics of computing~Distribution functions}

\keywords{$\sharpp$-hardness, shortest path, longest path, random edge length, distribution function, local uniformity, treewidth, $\xp$} 

\category{} 

\relatedversion{} 

\EventEditors{John Q. Open and Joan R. Access}
\EventNoEds{2}
\EventLongTitle{42nd Conference on Very Important Topics (CVIT 2016)}
\EventShortTitle{CVIT 2016}
\EventAcronym{CVIT}
\EventYear{2016}
\EventDate{December 24--27, 2016}
\EventLocation{Little Whinging, United Kingdom}
\EventLogo{}
\SeriesVolume{42}
\ArticleNo{23}

\begin{document}

\maketitle

\begin{abstract}
We investigate the problem of computing the distribution function 
for the shortest and longest path lengths in a directed graph with random edge lengths. 
Specifically, when these lengths are uniformly distributed, 
the problem reduces to computing the volume of a polytope defined by the graph structure. 
We establish that the problem is 
$\sharpp$-hard, even under the restricted condition that the random edge lengths 
are identically and independently distributed (i.i.d.) according to any continuous 
probability distribution with certain natural conditions, the local uniformity. 
This hardness result applies broadly: while the uniform distribution provides 
an essential case for the reduction, other distributions---such 
as exponential or normal---are similarly hard because they contain uniform 
distributions in every arbitrarily small interval. 
Furthermore, we show that the problem is contained within 
$\xp$ with respect to the treewidth $k$ 
of the underlying undirected graph. 
For the specific case of i.i.d. uniform edge lengths, 
we present a novel dynamic programming algorithm that processes 
a tree decomposition by iteratively performing convolutions 
to propagate distribution functions. 
Our approach achieves a time complexity of 
$n^{O(k^2)}$ for any fixed treewidth $k$.
\end{abstract}

\section{Introduction}
\noindent{\bf Motivation and Background.}
Computing the distribution function of the longest path length with random edge lengths 
presents a fundamental challenge in Statistical Static Timing Analysis (SSTA) for 
VLSI design. In modern semiconductor manufacturing, timing violations arising from 
process variations are a primary source of yield loss and carry significant economic
consequences. Since deterministic analysis fails to capture these stochastic 
behaviors due to constant scaling in semiconductor devices,
the performance of a logic circuit should be modeled using its 
full path length distribution. 
While numerous approximation algorithms have been proposed assuming continuous 
delays~\cite{BCSS2008,DK2003,VRKWN2004}, and the necessity of such stochastic modeling 
is well-documented~\cite{Wagner2016}, the problem remains formidable. Indeed, even 
modern statistical approaches, including machine learning methods, are actively being 
explored, as noted in recent surveys like~\cite{KO2023}.

The computational challenge predates the VLSI era, finding its roots in 
Operations Research with the Performance Evaluation and Review Technique (PERT), 
introduced by Malcolm, Roseboom, Clark, and Fazar~\cite{MRCF1959}. 
PERT models project duration as the longest path length within a Directed Acyclic 
Graph (DAG), where individual task durations are modeled using random variables. 
Over the decades, significant scholarly attention has been dedicated to determining 
the distribution functions of both shortest and longest paths, resulting in various 
heuristic and bounding methods (e.g., Adlakha and Kulkarni~\cite{AK1989}; 
Ludwig, Möhring, and Stork~\cite{LMS2001}). 

The computational complexity of these path length distribution problems is rooted 
in the class $\sharpp$, which was first defined by Valiant~\cite{Valiant1979tcs} 
to characterize the intractability in counting and reliability problems \cite{Valiant1979siam}. 
Initial work on stochastic networks demonstrated that computing both shortest and 
longest path length distributions is $\np$-hard (Ball and Provan~\cite{BP1983}). 
Subsequently, Hagstrom~\cite{Hagstrom1988} established the 
$\sharpp$-completeness of the problem when edge lengths are restricted to 
discrete random variables. 
There the reliability problem was used in the reduction,
associating the discrete edge lengths up (operational) or down (failure). 
However, this cannot be applied to the continuous edge lengths.

When random edge lengths are continuously uniformly distributed, 
the longest path length distribution function is equivalent to computing the volume of 
a high-dimensional polytope. Dyer and Frieze~\cite{DF1988} previously demonstrated the 
$\sharpp$-hardness of computing the distribution of sums of mutually independent 
uniformly random variables, a problem corresponding to the volume of a hypercube 
truncated by a halfspace. 
Their result proves hardness for cases where edge lengths are not necessarily
identically distributed. 
Furthermore, the fundamental difficulty extends 
beyond exact computation: Elekes~\cite{Elekes1986} 
proved that any polynomial-time algorithm cannot achieve a large approximation ratio 
for the volume of a general convex body 
(B\'{a}r\'{a}ny and F\"{u}redi provide an even larger bound in \cite{BF1987}).

Despite these foundational results, complexity studies on shortest and longest path 
length distributions have remained sparse until more recently. Relatively recently, 
Nouri and Ghodsi~\cite{NG2012} considered a stochastic task scheduling problem with
exponentially distributed random durations, proving that minimizing the expected 
makespan is $\np$-hard. Their model aligns with our problem when the workers are 
unlimited; however, we address the more demanding task of computing the entire 
distribution function of the longest path length.

Drawing from the argument above, a significant gap persists between theoretical 
complexity analyses and the practical requirements of the semiconductor industry. 
While $\sharpp$-completeness has been established for discrete edge lengths, 
real-world industrial applications require efficient algorithms that handle 
continuous edge distributions. However, proving that computing 
the distribution functions of the shortest/longest path length is 
$\sharpp$-hard for i.i.d. continuous edge lengths has remained an open 
challenge for many years. This required the development of new techniques capable 
of connecting established discrete $\sharpp$-complete problems to the volume 
defined by polytopes in a continuous space. 

The foundational breakthrough was the 
$\sharpp$-completeness proof of $\#$LE2H by Dittmer and Pak~\cite{DP2020}. 
They demonstrated that computing the number of linear extensions ($\#$LE) 
for a height-two poset given by $\mathcal{P}=(U\cup V, R)~~(R\subseteq U\times V)$
is $\sharpp$-complete.
Our approach to proving the $\sharpp$-hardness for computing shortest and 
longest path length distribution functions with continuous random variables 
leverages the established hardness of calculating the volume of the {\em order polytope}.
Specifically, Stanley~\cite{Stanley1986} showed that the order polytope's volume 
is equivalent to the number of linear extensions of its underlying poset. 
Brightwell and Winkler~\cite{BW1991} initially proved that counting linear 
extensions is $\sharpp$-complete for posets of height at least three.
Then, Dittmer and Pak~\cite{DP2020} proved the hardness 
for posets of height two after almost three decades.

The simplicity of $\#$LE2H provides a powerful framework for our reduction.
For our proof, we consider an associated DAG 
$G=(\{s,t\}\cup U\cup V,R\cup P)$, corresponding to $\mathcal{P}=(U\cup V,R)$. 
Here, the edge set $P$ consists of the edges from $s$ to each vertices in $U$, 
and edges from each vertices in $V$ to $t$.
We then assign static length $0$ to edges in $R$, and 
random lengths i.i.d. uniformly over $[0,1]$ to the edges in $P$. 
By construction, we prove that the shortest/longest path length distribution 
function coincides with the volume of the order polytope defined by $\mathcal{P}$, 
establishing our starting point for the $\sharpp$-hardness proof. 
We call $G$ a {\em transportation graph}
(see Fig.\ref{fig:reduction}(left) in page \pageref{fig:reduction}).

Our next key concept is the use of repeated convolution 
representation for the shortest/longest path length distribution function. 
Early work by Ando, Ono, Sadakane, and Yamashita~\cite{AOSY2009} demonstrated that, 
for a DAG, the longest path length distribution function (LPPDF) could be 
represented by a sequence of $O(n)$ convolutions. They utilized the 
formula to prove the tractability of LPPDF computation in DAGs with bounded ''width.''
Subsequently, Ando~\cite{Ando2017} utilized the 
pathwidth of the underlying undirected graph of the DAG. 
There, the existence of a Fully Polynomial-Time Approximation Scheme 
(FPTAS) was demonstrated 
if the pathwidth is bounded by a constant and edge lengths are mutually independent 
(not necessarily identical) uniformly distributed random variables. 

While these advances established significant results, they often rely on specialized 
constraints or distributions. 
Previous research has generally been constrained to 
specific probability models probably due to the hardness of 
deriving general representations. 
See e.g., \cite{NKBM2006} (normal distribution),
\cite{Kulkarni1986,NG2012} (exponential distribution).
On the contrary, we can concentrate on the local uniformity of many natural distributions
by the above repeated convolution representation of LPPDF, 
which leads to the $\sharpp$-hardness results for many distributions at a time.

For our tractability result,
we utilize {\em treewidth} as our primary structural parameter because 
it represents a broader class of tractable graphs compared to the various 
"width" parameters used by previous works \cite{AOSY2009, NG2012} 
(e.g., Dilworth~\cite{Dilworth1950}, M\"{o}hring~\cite{Moehring1989}). 
Introduced by Robertson and Seymour~\cite{RS1986}, treewidth quantifies how 
closely a graph approximates a tree structure. 
The utility of the measure was significantly advanced by Bodlaender’s 
algorithm~\cite{Bodlaender1996}. 
Courcelle's Theorem and subsequent work 
\cite{CourcelleEngelfriet} demonstrated that any problem expressible in Monadic Second 
Order Logic (MSOL) is solvable in linear time on graphs with bounded treewidth. 
In addition, Arnborg and Lagergren~\cite{AL1991} extended the capability to optimization 
problems involving lengths. 
However, calculating a probability distribution 
function does not fit directly into the MSOL framework
since it requires iterative convolutions.

To bridge the gap, 
we present a polynomial-time algorithm for computing both shortest and longest 
path length distribution functions on graphs of bounded underlying treewidth. 
This result shows that our problems lie within the complexity class $\xp$.

\bigskip
\noindent {\bf Main Contribution.}
We consider the following problem. 
Given a graph $G=(V,E)$ and random edge lengths $\Vec{X}\in \mathbb{R}^{E}$, 
let $X_{\rm MIN}$ and $X_{\rm MAX}$ be the minimum and maximum path length, respectively. 
Specifically, we define 
$X_{\rm MIN}=\min_{\pi\in\Pi_G} \left\{\sum_{e\in \pi} X_e\right\}$ 
and $X_{\rm MAX}=\max_{\pi\in\Pi_G}\left\{\sum_{e\in \pi}X_e\right\}$, 
where $\Pi_G$ is the set of all simple paths from source $s$ to sink $t$.

In either case, our goal is to compute the distribution functions:  
$F_{\rm MIN}^{\Vec{X}}(x)=\Pr[X_{\rm MIN}\le x]$ and 
$F_{\rm MAX}^{\Vec{X}}(x)=\Pr[X_{\rm MAX}\le x]$. 
When all components of $\Vec{X}$ are independent and identically distributed (i.i.d.), 
we refer to these respective problems as SPPDF-IID and LPPDF-IID.
Our main contributions are summarized as follows.
\begin{enumerate}
\item $\sharpp$-hardness Proof for i.i.d. Continuous Uniform Edge Lengths: 
We establish the $\sharpp$-hardness of SPPDF-IID and LPPDF-IID. 
Extending the results in~\cite{DF1988}, 
we prove that the intractability persists even for i.i.d. edge lengths. 
Our initial result establishes a continuous analogue to Hagstrom~\cite{Hagstrom1988},
which originally restricted its scope to directed graphs. 
Furthermore, we generalize this result to undirected graphs.
\begin{restatable}{theorem}{thDirectedHard}
\label{th:directedhard}
For a directed graph, SPPDF-IID and LPPDF-IID are both $\sharpp$-hard if
the edge lengths are uniformly distributed over $[0,1]$. 
\end{restatable}
Subsequently, we extend our analysis to the undirected case:
\begin{restatable}{theorem}{thUndirectedHard}
\label{th:undirectedhard}
For an undirected graph, SPPDF-IID and LPPDF-IID are both 
$\sharpp$-hard if the edge lengths are uniformly distributed 
over $[1,2]$, or over $[-2,-1]$.
\end{restatable}

\item $\sharpp$-hardness Proof for Edge Length Distributions with Local Uniformity:
We restrict our focus to distributions belonging to the class defined as follows.
\begin{definition}
\label{definition:localuniformity}
\noindent {\bf Local Uniformity}:
For $I,J\subseteq\mathbb{R}$ and $\ell\in \mathbb{R}$,
a polynomial time computable distribution function $F(x)$ 
belongs to $\distfamily_\ell(I;J)$ if
(i) for any $x\in I$, $\frac{d}{dx}F(x)=f(x)> \ell$, 
(ii)$\left|\frac{d}{dx} f(x)\right|$ is bounded for $x\in I$, and
(iii) $f(x)=0$ for all $x\in J$.
\end{definition}
The key implication that
the local uniformity allows us to magnify $f(x)$ to approximate uniformness. 
Such an interval $I$ is commonly found among natural probability distributions. 
Since we require$f(x)=0$ for $x\in J$, 
we define the following for concise description.
\begin{definition}
Let $I\subseteq \mathbb{R}$.
We denote by $I_{\leftarrow}$ (resp. $I_{\rightarrow}$) a one-dimensional 
halfspace $\{x\in\mathbb{R}|x< y,~\forall y\in I\}$ 
(resp. $\{x\in\mathbb{R}|x> y,~\forall y\in I\}$).
\end{definition}
We employ the families $\distfamily_{\ell}(I;I_{\leftarrow})$ and
$\distfamily_{\ell}(I;I_{\rightarrow})$
to perform the magnification process, leading to our next theorem:
\begin{restatable}{theorem}{thDirectedHardManyDistributions}
\label{th:directedhard_manydistributions}
For a directed graph with $n$ vertices, SPPDF-IID and LPPDF-IID are both $\sharpp$-hard if
the distribution of the edge lengths is in $\distfamily_{\ell}(I;\emptyset)$
where $I=[x_0,x_0+\epsilon]$
for $\exists x_0\in \mathbb{R}$, $\ell=\Omega(2^{-\poly(n)})$ 
and $\epsilon=O(\ell/(n+1)!)$.
\end{restatable}
We show some examples of the distributions that our $\sharpp$-hardenss applies
as follows.
\begin{corollary}
For directed graphs, SPPDF-IID and LPPDF-IID are 
$\sharpp$-hard when the edge lengths are drawn from any of the following distributions: 
uniform, exponential, gamma, normal, or Cauchy.
\end{corollary}

We extend the above result to the setting of undirected graphs
by incorporating parallel edges.
We observe that the resulting distribution class is somewhat 
constrained because our reduction relies on how we can separate
the support of the density from the origin.
\begin{restatable}{theorem}{thUndirectedHardManyDistributions}
\label{th:undirectedhard_manydistributions}
For an undirected graph with $n$ vertices, let $F(x)$ be the distribution 
function of the i.i.d. random edge lengths.
Let $I^+=[x_0,x_0+\epsilon]$ and $I^-=[-x_0-\epsilon,-x_0]$ 
for $\exists x_0\ge \epsilon$,
$\epsilon=O(\ell/(n+1)!)$ and $\ell=\Omega(2^{-\poly(n)})$. 
Also, let $x_\ell$ satisfy $\epsilon\le x_\ell \le x_0$.
Then, LPPDF-IID is $\sharpp$-hard if
(i) $F(x_\ell)\le 1-1/\poly(n)$, and 
(ii) there exists an $N=O(\poly(n))$ such that 
$F^N\in \distfamily_{\ell}(I^+;\emptyset)\cup\distfamily_{\ell}(I^-;I^-_\rightarrow)$.
Similarly, SPPDF-IID is $\sharpp$-hard if
(i) $F(-x_\ell)\ge 1/\poly(n)$, and
(ii) the distribution function $1-\bar F^N(x)$ of the minimum of $N$ parallel edge lengths 
belongs to $\distfamily_{\ell}(I^-;\emptyset)\cup \distfamily_{\ell}(I^+;I^+_\leftarrow)$, where $\bar F(x)=1-F(x)$.
\end{restatable}
For LPPDF-IID, 
the distribution function class 
$\distfamily_{\ell}(I^+;\emptyset)\cup \distfamily_{\ell}(I^-;I^-_\rightarrow)$ 
for $x_0\ge \epsilon$ represents distributions
whose maximum of polynomially many edge lengths 
has local uniformity in $I^+$ or $I^-$, and
we can separate the support from the $\epsilon$-neighborhood of the origin
by replacing each edge by $N$ parallel edges.
The condition $F(x_\ell)\le 1-1/\poly(n)$ is necessary for ignoring the probability
for $x<x_\ell$ by the parallel edges.
In the union, the earlier $\distfamily_{\ell}(I^+;\emptyset)$ 
and the latter $\distfamily_{\ell}(I^-;I^-_\rightarrow)$ are not symmetry
to each other since the parallel edge pushes the random edge lengths toward 
the positive direction in LPPDF-IID.
We have a symmetry analogue for SPPDF-IID. 
Note that the parallel edge works for the minimum of edge lengths, which
pushes the random edge lengths to the negative direction. 
To show some well-known distributions, we state the following.
\begin{corollary}
For an undirected graph $G=(V,E)$ with $n$ vertices,
SPPDF-IID (resp. LPPDF-IID) is $\sharpp$-hard 
if the length of each edge is $X_e-x_0$ (resp. $X+x_0$) 
for $x_0\ge 2^{-\poly(n)}$
where $X_e$ for $e\in E$ is a random variable drawn from
uniform over $[0,1]$, exponential, gamma, normal or Cauchy.
\end{corollary}
It is not very likely that the translation $x_0$ is essential.
We conjecture as follows.
\begin{conjecture}
For an undirected graph with $n$ vertices, 
SPPDF-IID and LPPDF-IID are both $\sharpp$-hard if
the distribution of the edge lengths is in $\distfamily_{\ell}(I;\emptyset)$
where $I=[x_0,x_0+\epsilon]$
for $\exists x_0\in \mathbb{R}$, $\ell=\Omega(2^{-\poly(n)})$ 
and $\epsilon=O(\ell/(n+1)!)$.
\end{conjecture}

\item Exact $\xp$ Algorithm with Respect to Treewidth:
We present an algorithm that computes the exact distribution functions for both shortest and longest path lengths, specifically when edge lengths are i.i.d. uniform random variables. Our key result is summarized in Theorem \ref{th:XPalgorithm}:
\begin{restatable}{theorem}{thXPalgorithm}
\label{th:XPalgorithm}
  Given a directed graph $G=(V,E)$ and its tree decomposition with width $k$,
  we can solve SPPDF-IID and LPPDF-IID in
  $O(n^{2k^2+2k+2})$ time if the edge lengths obey the uniform distribution
  over $[0,1]$ or $[-1,0]$.
\end{restatable}
The result confirms that the presence of cycles within a graph does not 
significantly increase the problem's hardness with respect to $\xp$. 
\end{enumerate}
\smallskip
\noindent{\bf Organization.}
The remainder of this paper is organized as follows.
Section 2 presents the overview of our techniques.
Section 3 shows detailed notations and formalizes 
the shortest and longest path problems. 
Section 4 establishes the $\sharpp$-hardness for i.i.d. continuous edge lengths,
proving Theorems~\ref{th:directedhard}, \ref{th:undirectedhard}, \ref{th:directedhard_manydistributions}, and \ref{th:undirectedhard_manydistributions}.
In Section 5, we present the $\xp$ algorithm for i.i.d. uniform edge lengths,
where we prove Theorem~\ref{th:XPalgorithm}.
Section 6 shows the conclusion. 
In addition, we have Section 7 for supplemental proofs: 
Propositions~\ref{proposition:convolution}, \ref{proposition:PsiMax}, \ref{proposition:constant_min}, and Theorem~\ref{th:DAGLP}.

\section{Technical Overview}
\subsection{$\sharpp$-hardness for Directed/Undirected Graphs with i.i.d. Uniformly Distributed Edge Lengths}
For our $\sharpp$-hardness reduction concerning Theorem~\ref{th:directedhard}, 
we restrict our focus to the interval $0\le x\le 1$ for $F_{\rm MAX}(x)$. 
A crucial property of the uniform distribution is that when edge lengths are 
i.i.d. uniformly distributed over $[0,1]$, 
$F_{\rm MAX}(x)$ takes the form of a monomial over the interval. 
Furthermore, the monomiality is preserved even if 
we constrain some edge lengths to fixed value zero.
\begin{proposition}
\label{proposition:constant}
Let $G=(V,E)$ be a graph.
We assume that there are $m_0$ edges with static length $0$ and
$m_1=|E|-m_0$ edges with uniform random lengths i.i.d. over $[0,1]$.
For $0\le x \le 1$, we have $F_{\rm MAX}(x)=Ax^{m_1}$ for a constant $A$. 
\end{proposition}
\begin{proof}
By definition, for $0\le x\le 1$, $F_{\rm MAX}(x)$ corresponds to the volume of 
an $m_1$-dimensional polytope $K(x)$. 
The facet defining constraints are given by $\Vec{x}\ge \Vec{0}$ and 
$\sum_{e\in \pi}x_e\le x$ for all $\pi \in \Pi_G$. 
Since every vertex position of $K(x)$ is scaled linearly by $x$, 
its volume is proportional to a monomial in $x$, 
implying the proposition.
\end{proof}
Thus, difficulty of computing the constant $A$ determines 
the $\sharpp$-hardness of the problem. Notably,
Proposition~\ref{proposition:constant} holds regardless of whether 
the graph is directed or undirected. 
While the placement of edges with random lengths affects the numerical 
constant factor, we focus on $F_{\rm MAX}(x)$ (LPPDF-IID) because 
the choice maintains a cleaner probability description. 
We will defer the introduction of the analogue for 
$F_{\rm MIN}(x)$ until Proposition~\ref{proposition:constant_min} 
on page \pageref{proposition:constant_min}. Instead, we remark
that the results for LPPDF-IID holds symmetrically for SPPDF-IID 
as long as we are concerned with simple paths.
\begin{restatable}{proposition}{proposition_minmaxsymmetry}
\label{proposition:minmaxsymmetry}
Consider a directed graph $G=(V,E)$ 
and its random edge length vector $\Vec{X}\in\mathbb{R}^E$.
For any $x\in\mathbb{R}$ satisfying $\Pr[X_e=x]=0$ for any $e\in E$,
we have $\bar F_{\rm MIN}^{\Vec{X}}(x)=F_{\rm MAX}^{-\Vec{X}}(-x)$. 
\end{restatable}
\begin{proof}
Remember that $\Pi_G$ is the set of all simple $s-t$ paths in $G$. 
By definition, we have
\begin{align*}
    \bar{F}_{\rm MIN}^{\Vec{X}}(x)=\Pr\left[\bigwedge_{\pi\in \Pi_G}\left(\sum_{e\in \pi} X_e > x\right)\right]=
    \Pr\left[\bigwedge_{\pi\in \Pi_G}\left(-\sum_{e\in \pi} X_e \le -x\right)\right]=F_{\rm MAX}^{-\Vec{X}}(-x).
\end{align*}
\end{proof}

Our $\sharpp$-hardness results are built upon the work of Dittmer and Pak~\cite{DP2020}. 
By incorporating both edges with random lengths and those with static lengths, 
we show that solving our problem on a corresponding transportation graph 
is equivalent to calculating $\#$LE2H, which is $\sharpp$-complete.
\begin{restatable}{definition}{definitionLEtwoH}
\label{definition:LE2H}
{\bf $\#$LE2H}: 
Given a height-2 poset $\mathcal{P}=(U\cup V,R)$, 
where $R$ is a partial order subset of $U\times V$. 
Compute the number of permutations on $U\cup V$ that 
do not contradict the partial order $R$.
\end{restatable}
\begin{restatable}{theorem}{thDP}
\cite{DP2020}
\label{th:DP2020}
$\#$LE2H is $\sharpp$-complete.
\end{restatable}
Furthermore, it is a known result that $\#$LE (the number of linear extensions) 
is proportional to the volume of the corresponding order polytope. 
This relationship is formalized by Stanley~\cite{Stanley1986}:
\begin{restatable}{theorem}{thorderpolytope}
\label{th:orderpolytope}
\cite{Stanley1986}
Let $\mathcal{P}=(U\cup V,R)~(R\subseteq U\times V)$ be a poset of height two.
Let $\mathcal{O_P}$ be a $|U\cup V|$-dimensional polytope, where
$\mathcal{O_P}=\{\Vec{x}\in[0,1]^{U\cup V}|x_u\le x_v ~~\text{for all}~~uv\in R\}$.
Then, ${\rm Vol}(\mathcal{O_P})=\frac{\#{\rm LE2H}}{(|U|+|V|)!}$.
\end{restatable}
For the initial part of our hardness results, the $\sharpp$-hardness of SPPDF and LPPDF 
with i.i.d. uniformly random length edges combined with fixed zero-length edges is 
formalized by Lemma~\ref{lemma:Uniform-hard}.
We use the concatenation $\cat{\Vec{u}}{\Vec{v}}$ of 
two vectors $\Vec{u}$ and $\Vec{v}$ in the statement.
\begin{restatable}{lemma}{lemmaUniformhard}
\label{lemma:Uniform-hard}
Given a height-2 poset $\mathcal{P}=(U\cup V, R)$ 
and its related transportation graph 
$\graphzero_{R}=(U\cup V\cup\{s,t\},R\cup P)$.
Let $\Vec{X}\in [0,1]^P$ be a random vector with components 
$X_e\iidsim \Unif$ for $e\in P$.
Computing 
$F_{\rm MAX}^{\cat{\Vec{0}}{\Vec{X}}}(1)$ and 
$\bar F_{\rm MIN}^{\cat{\Vec{0}}{\Vec{X}}}(1)$
exactly is $\sharpp$-hard.
\end{restatable}

Figure~\ref{fig:reduction} illustrates the transformation process 
for three instances sharing an underlying poset structure. 
We begin with an initial transportation graph (left) 
characterized by zero-length edges and progress through two subsequent forms: 
one where all edge lengths are i.i.d. uniformly over $[0,1]$ (center), 
and an intermediate incorporating the $|R|$-edge "tail" (right). 
Given that their LPPDFs are monomials for $x \le 1$ by
Proposition~\ref{proposition:constant}, 
we prove an equivalence relationship among all three instances, 
with details provided in Proposition~\ref{proposition:tail} and 
Lemma~\ref{lemma:equalPsi}. 
The fundamental technique involves
commutativity of convolution.
We show a proof of Proposition~\ref{proposition:convolution} in Section~\ref{section:supplementary},
since it provides the fundamental basis of Theorem~\ref{th:DAGLP}.
\begin{restatable}{proposition}{propositionConvolution}
\label{proposition:convolution}
\noindent {\bf Commutativity of Convolution:}
Let $X_1$ and $X_2$ be two mutually independent random variables,
where $F_i(x)=\Pr[X_i\le x]$ and $f_i(x)=\frac{d}{dx}F_i(x)~(i=1,2)$. Then, we have
\begin{align*}
\Pr[X_1+X_2\le x] = \int_{\mathbb{R}}f_1(x_1)F_2(x-x_1) dx_1= \int_{\mathbb{R}}F_1(x_1)f_2(x-x_1) dx_1.
\end{align*}
\end{restatable}
By applying the latter identity between two integrals twice to 
$F_{\rm MAX}^{\Vec{X}}(x)$'s integral representation which
we explain next, we shift the randomness from an 
incoming edge length to an outgoing edge length 
(for more detail, see also Fig.~\ref{fig:reduction}, \ref{fig:subdivision} and Lemma~\ref{lemma:equalPsi}). 
The manupilation is kept simple since
the $s-t$ paths in the transportation
graph has exactly three edges. 
\begin{figure}[ht]
\begin{center}
    \includegraphics[clip,scale=0.7, bb=10 620 470 850]{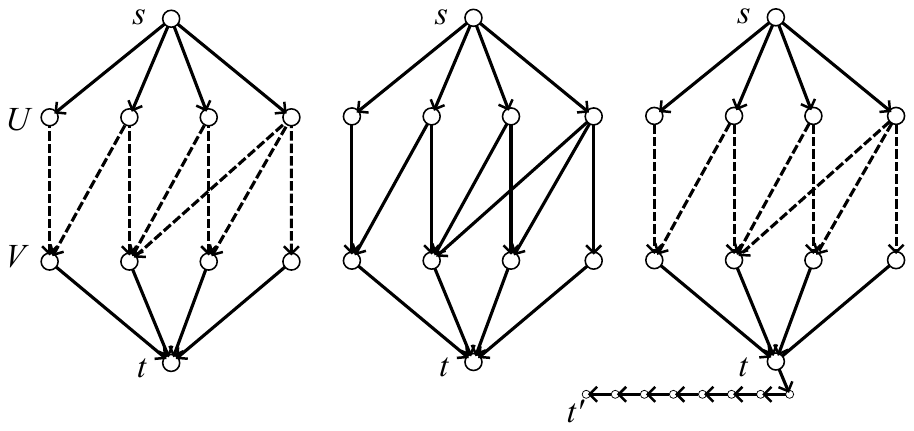}
\end{center}
\vspace*{-10pt}
\caption{
On the left, we establish that for the transportation graph $\graphzero_R$, 
the shortest/longest path distribution function equals the volume of an associated 
order polytope when edges in $R \subseteq U \times V$ are zero-length and 
other edge lengths are i.i.d. uniform on $[0,1]$. We then analyze a transformation 
(center to right): starting with an entirely i.i.d. random length edge graph, 
we introduce a ``tail'' containing exactly $|R|$ edges. 
}
\label{fig:reduction}
\end{figure}
The manipulation stands on 
Theorem~\ref{th:DAGLP} established in \cite{AOSY2009}.
A revised version of the proof is provided 
in Section~\ref{section:supplementary} for self-containment.
\begin{restatable}{theorem}{thDAGLP}
\cite{AOSY2009}
\label{th:DAGLP}
Let $G=(V,E)$ be a directed acyclic graph (DAG) with vertex set $V$ and edge set $E$. 
Let $S$ and $T$ denote the sets of sources and sinks, respectively. 
Let $\Vec{X}\in\mathbb{R}^E$ 
be any random vector whose components are mutually independent, 
but not necessarily identically distributed.
For an edge $uv \in E$, let $X_{uv}$ be its random length, 
and let $F_{uv}(x)=\Pr[X_{uv}\le x]$ be the corresponding distribution function. 
Then, we have 
\begin{align*}
F_{\rm MAX}^{\Vec{X}}(x)=\int\limits_{\mathbb{R}^{V\setminus T}} \prod_{s\in S} H(x-z_s) \prod_{u\in V\setminus T} \frac{\partial}{\partial z_u}\prod_{v\in V_u} F_{uv}(z_u-z_v) {\rm d}\Vec{z},
\end{align*}
where $V_u=\{v\in V| uv\in E\}$ is the set of children of vertex $u$, and 
we set $z_t=0$ for all sinks $t\in T$. 
\end{restatable}
Moreover, the explicit formulation of $F_{\rm MAX}^{\Vec{X}}(x)$ 
reveals how the structure of a DAG can be systematically decomposed 
and reassembled. By decomposing the DAG corresponding to an instance of $\#$LE2H, 
we reduce LPPDF to computing the multivariate convolution of 
two readily computable functions. 
These outcomes are shown in 
Propositions~\ref{proposition:PsiMax}, 
\ref{proposition:psi_U}, and \ref{proposition:psi_V}.

Extending our results to undirected graphs, 
there is another challenge, 
specifically due to candidate paths having varying numbers of edges. 
We address this issue using two techniques. 
First, by considering the uniform distribution over $[1,2]$ (rather than $[0,1]$), 
we separate the path length density function based on edge count. 
Second, we add a cycle that exhaustively traverses the neighbor of the source of
the transportation graph (see, Fig.~\ref{fig:Undirected_Reduction}). 
The modification guarantees an identical 
edge count for all candidate longest paths. Taken together, 
these techniques allow us to prove that $F_{\rm MAX}(x)$ is 
computationally hard in its leftmost support interval.

\subsection{$\sharpp$-hardness of Graphs with i.i.d. Edge Lengths Drawn From
Any Distribution with Local Uniformity}
To generalize the result to distributions with local uniformity, 
we approximate any edge length density function $f(x)$ by  
bounding $x$ so that we can utilize the almost-uniform behavior. 
We present three steps as follows.
\begin{enumerate}[{\normalfont (\roman*)}]
\item Suppose we have an edge length density $f(x)$ with a jump 
on the boundary of its support (e.g., the exponential distribution). 
By focusing on the proximity of the jump, 
it can be shown that $f(x)$ is approximately a uniform distribution. 
To be precise, consider the following.
Given that $X_1$ and $X_2$ are non-negative random variables with density $f(x)$, 
calculating the probability $\Pr[X_1+X_2\le x]$ simplifies 
to considering only the case $X_1,X_2\in [0,x]$. 
For sufficiently small $x$, this constraint, combined 
with the jump in $f(x)$ when $X_1$ and $X_2$ are i.i.d.,
which demonstrates that $f(x)$ behaves as a uniform density near zero.
This step is presented as Lemma~\ref{lemma:Breakpoint-hard}.
\item If $f(x)$ does not have any jump at its support boundary
but bounded from below by $a$ (e.g., gamma distribution with shape parameter $\alpha >1$),
we translate $f(x-a)$ so that the support is nonnegative.
Then, since the probability calculation is restricted to the
path lengths at most $x_0$, 
we truncate the function $f(x-a)$ at $x=x_0+a$.
This truncation introduces an artificial discontinuity, 
leading to an approximate uniform distribution. 
Now, $f(-(x-x_0-a))$ applies symmetrically to step (i).
We have the case as Lemma~\ref{lemma:Breakpoint-hard-interval}.
\item Even for unbounded supports (e.g., normal distribution), 
the density $f(x)$ must decay toward zero as 
$x \rightarrow \pm\infty$, since $\int_{\mathbb{R}} f(x) dx = 1$. 
The distribution $F(x)$ is negligible
below a certain point $x_\ell$ 
(i.e., $F(x)\le \ell=\Omega(2^{-\poly(n)})$ for $x\le x_\ell$).
We truncate $f(x)$ at $x_\ell$, leading to 
an approximate uniform distribution near some $x_0 \ge x_\ell$, 
which falls upon the previous step (ii) (Lemma~\ref{lemma:infty_hard}). 
\end{enumerate}

While extending the $\sharpp$-hardness result to undirected graphs 
remains conceptually similar,  
the edge length density $f(x)$ must be separable from the origin,
since we use the uniform distribution over $[1,2]$ for our reduction.
We address two cases: First, 
if the support of $f(x)$ is strictly confined to one side 
of the origin with some margin 
(i.e., $f(x)=0$ for $|x|\le \epsilon$), 
we scale the neighborhood of the origin by $1/\epsilon$. 
It allows us to approximate uniform density over $[1,2]$ by $f(x/\epsilon)$.
Second, if $f(x)$ includes the origin's neighborhood, 
we introduce polynomial number of parallel edges to impose biases,
positive for LPPDF-IID ($\max$ operation) and 
negative for SPPDF-IID ($\min$ operation). 
We see that the resulting distribution of parallel edge lengths 
is negligible at $x_\ell=0$, as established in step (iii) of the above.

\subsection{$\xp$ Algorithm for Directed Graphs with Edge Lengths i.i.d. Uniformly over $[0,1]$}
Our $\xp$ algorithm generalizes prior results \cite{AOSY2009,Ando2017} 
by accommodating graphs with cycles. 
Our algorithm is a bottom-up dynamic programming across a tree decomposition, 
maintaining state as a piecewise polynomial. 
By developing a recurrence utilizing local convolutions, 
we efficiently propagate the distribution function from 
the leaf bags to the root. 
Let $G=(V,E)$ be a directed graph.
Specifically, our $\xp$ algorithm advances upon the prior work 
through three tricks.
\begin{enumerate}[{\normalfont (\arabic*)}]
\item Firstly, we introduce two dummy variables $z_{uv,u}$ and $z_{uv,v}$
for every edge $uv\in E$, 
contrasting with the single variable per vertex used previously. 
The modification provides a degree of redundancy when applied to DAGs, 
which at the same time grants the necessary capability to handle 
graphs containing cycles.

\item Secondly, we define two gadgets,  distinctly for vertices and edges, 
whereas prior research treated a vertex and its outgoing edges as 
an indivisible unit. 
Here, a gadget is a formula that describes probability of some local event.
Our two variables for each edge are defined to represent the optimal path length 
starting from either the source or the destination of each edge, ending at the 
overall graph sink.
By the description of our construction, we have more clarity relative to 
the prior work. 

\item Thirdly, we explicitly encode local convolutions 
using gadgets for both vertices and edges. 
By performing partial differentiation of a gadget, 
we obtain a probability density corresponding to that one inequality 
in the probability condition turned into a equality.
A convolution with the partial derivative process combines two gadgets, 
realizing substituting one variable in the probability condition 
by the same variable in the above equality. 
Finally, the distribution function for the shortest/longest path length 
is constructed by properly taking the product or convolution of 
these vertex and edge probability descriptions.
\end{enumerate}

We here have a closer look at our algorithm.
Interestingly, we simplify the analysis by restricting consideration to 
SPPDF-IID, or $F_{\rm MIN}(x)$ in contrast to that 
considering $F_{\rm MAX}(x)$ saved the description in the previous Subsections. 
Since our algorithm can be applied to the edge lengths uniformly distributed over $[-1,0]$,
too, our algorithm can be applied to $F_{\rm MAX}(x)$ by Proposition~\ref{proposition:minmaxsymmetry}.
We further simplify the argument by introducing 
{\em sentinel vertex} ($\bot$) and 
its adjacent edges, {\em sentinel edges}. 
We fix the edge lengths associated with these sentinel edges. 
\begin{restatable}{definition}{definitionsentinel}
\label{definition:sentinel}
Given a directed graph $G=(V,E)~(V=\{1,\dots,n\})$, 
we add a sentinel vertex $\bot$ and
sentinel edges $E_\bot=\{\bot v, v\bot | v\in \{1,\dots,n\}\}$. 
For $s\in S(G)$ and $t\in T(G)$, 
the lengths of $\bot s$ and $t\bot$ are $0$;
otherwise, the edge length is a sufficiently large static value $M$.
\end{restatable}
The use of sentinel vertex saves the algorithm description, 
allowing the algorithm's core operation to be restricted solely 
to the sequential recurrence across the vertices from the leaf 
bags to the root. Our algorithm is listed as follows:
\begin{restatable}{algorithm}{algorithmmainalg}
\label{algorithm:mainalg}
Input: Directed Graph $G=(V,E)$ and its Tree Decomposition $\mathcal{T}=(\mathcal{B},\mathcal{A})$\\
1. Order the vertices in the bottom-up order, from vertices in the leaf bags to the root. \\
2. Add sentinel vertex $\bot$ and sentinel edges $\bot v$ and $v\bot$ for each $v\in V$; \\
3. for $i=1,\dots,n$ do:\\
4. \hspace*{5mm} Compute $F_{\partial G_i}(\Vec{z}_{\partial_E[i]})$, as described in Section \ref{section:recurrence};\\
5. Output $F_{\partial G_n}(\Vec{z}_{\partial_E[n]})=F_{\rm MIN}(x)$.
\end{restatable}

The notation $\partial$ refers to the frontier in the computation. 
Specifically, let $G = (V,E)$ be the graph with vertex set $V=\{1,\dots,n\}$. 
We define $\partial_E[i]$ as the frontier edge set of 
the subgraph induced by the vertices $[i] = \{1, \dots, i\}$. 
The intermediate graph $\partial G_i$ is obtained 
by setting each $e=uv\in \partial_E[i]$ as a source or a sink vertex
depending on whether $u\not\in [i]$ and $v\in[i]$ ($e$ is a new source) or 
$u\in [i]$ and $v\not\in [i]$ ($e$ is a new sink). 
\begin{restatable}{definition}{definitiongreatergraph}
\label{definition:greatergraph}
Let $G_i=([i],E_i)$ be the subgraph of $G=(V,E)$ induced by $[i]$.
Let $\partial G_i$ be the graph with vertex set $[i]\cup \partial_E[i]$
and edge set 
$E_i\cup \{ve| v\in [i], e\in \partial_E^+\}\cup \{ev| e\in\partial_E^-, v\in[i]\}$.
We define $\Pi_i$ as the set of intersections between an $s-t$ path of $G$ and
paths between vertices in $S(\partial G_i)$ and $T(\partial G_i)$.
\end{restatable}
Our algorithm iterates over all vertices, excluding the sentinel 
vertex $\bot$. By setting the vector $\Vec{z}_{\partial_E[i]}$ 
appropriately, our algorithm outputs $F_{\rm MIN}(x)$.
\begin{restatable}{lemma}{lemmaFDN}
  \label{lemma:FDN}
  Let $\Vec{z}_{\partial_E[n]}$ is a vector indexed by 
  $\partial_E[\bot]=\{\bot v, v\bot| v\in V\}$.
  We set $z_{\bot s, s}=x$ and $z_{s \bot,\bot}=0$ for $s\in S(G)$; 
  $z_{\bot v,v}=z_{v,\bot,\bot}=0$ for $v\in V\setminus S(G)$ then, we have
  \begin{align}
    &F_{\partial G_n}(\Vec{z}_{\partial_E[n]})=\Pr\left[\min_{\pi\in\Pi_n}\left\{\sum_{e\in \pi}X_e\right\}\le x\right]=\Pr\left[\bigvee_{\pi\in\Pi_n}\left(\sum_{e\in \pi}X_e\le x\right)\right].\nonumber
  \end{align}
\end{restatable}
Specifically, $F_{\partial G_i}(\Vec{z}_{\partial_E[i]})$, appeared in Step 4 
of Algorithm~\ref{algorithm:mainalg}, 
is an intermediate distribution function. 
It takes the form of a piecewise polynomial parameterized 
by the components of $\Vec{z}$. 
Such functions can be encoded \LaTeX-like format and stored on the memory 
using the Heaviside step function $H(x)$ to describe 
case dependencies (e.g., the uniform distribution function is $\Unif(x)=x H(x)H(1-x)+ H(x-1)$
ignoring the value at $x=1$). 
Here, $F_{\partial G_i}(\Vec{z}_{\partial_E[i]})$ is the probability that  
$s-t$ paths in $\partial G_i$ satisfy path length constraints. 
These constraints are specified by the vector $\Vec{z}_{\partial_E[i]}$, 
which contains $|\partial_E[i]|$ components, 
corresponding precisely to the frontier edges of $\partial G_i$.
\begin{restatable}{definition}{intermediateF}
\label{definition:intermediateF}
Let $e_s(\pi)$ and $e_t(\pi)$ are the first and the last edges of a path $\pi$, respectively.
Also, let $s_\pi$ and $t_\pi$ are the source and the sink of $\pi$, respectively.
For $x\in\mathbb{R}$ and 
$\Vec{z}_{\partial_E[i]}\in\mathbb{R}^{\partial_E[i]}$,
we set 
\begin{align*} 
  F_{\partial G_i}(\Vec{z}_{\partial_E[i]})\!&=\!\Pr\left[\bigvee_{\pi\in \Pi_i}\bigvee_{p\in \mathrm{Con}(\pi)}\left(\sum_{e\in p} Z_e\le z_{e_s(p),s_p}-z_{e_t(p),t_p}\right)
  \right],
\end{align*}
where $\mathrm{Con}(\pi)$ is the set of connected components of $\pi$.
\end{restatable}
We use $\mathrm{Con}(\pi)$ since an $s-t$ path in $\Pi_G$ may enter and exit $\partial G_i$ 
multiple times.

Our algorithm must track an exponentially large number of paths in $\partial G_i$. 
Accordingly, we construct the local convolution using the following gadgets 
as atomic elements. In the following definition, $E^+(i)$ (resp. $E^-(i)$) is 
the set of outgoing (resp. incoming) edge set of vertex $i$.
\begin{restatable}{definition}{definitiongadgets}
\label{definition:gadgets}
Let $D_{\Vec{z}^-}=\prod_{ui\in E^-(i)}\frac{\partial}{\partial z_{ui,i}}$. 
For $\Vec{z}_{\rm src}^+\in \mathbb{R}^{E^+(i)}$ and $\Vec{z}^-_{\rm dst}\in\mathbb{R}^{E^-(i)}$,
we define the vertex gadget by
\begin{align*}
  &\mathcal{V}_i(\Vec{z}^+_{\rm src},\Vec{z}^-_{\rm dst})=D_{\Vec{z}^-}\Pr\left[\bigvee_{ui,iv\in E(i)}\mathcal{U}_{ui,iv}(\Vec{z}^+_{\rm src},\Vec{z}^-_{\rm dst})\right], \\
  &\text{where}~~\mathcal{U}_{ui,iv}(\Vec{z}^+_{\rm src},\Vec{z}^-_{\rm dst})=(0\le z_{ui,i} - z_{iv,i} )\wedge\mathrm{Garbage}(ui,iv,\Vec{z}),\\
  &\text{and}~~\mathrm{Garbage}(ui,iv,\Vec{z})=\bigwedge_{\begin{subarray}{c}
  u'i,iv'\in E(i)\\
  u'\neq u, v'\neq v
  \end{subarray}}(0\le z_{u'i,u'}-z_{i\bot,i}\wedge 0\le  z_{\bot i,i}-z_{iv',i}). 
\end{align*}
For $\Vec{z}_{\rm src},\Vec{z}_{\rm dst}\in \mathbb{R}^{E^+(i)}$, the edge gadget is 
\begin{align*}
    \mathcal{E}_i(\Vec{z}_{\rm src}, \Vec{z}_{\rm dst})&=\left(\prod_{iv\in E^+(i)}\frac{\partial}{\partial z_{iv,i}}\right)\Pr\left[\bigvee_{iv\in E^+(i)}\left(\vphantom{\frac{1}{1}}X_{iv}\le z_{iv,i}-z_{iv,v}\right)\right]\\
    &=\prod_{iv\in E^+(i)}\unif(z_{iv,i}-z_{iv,v}).
\end{align*}
\end{restatable}
See Remark~\ref{remark:generalizedfunction} for the shape of $\mathcal{V}_i$.
The edge gadget $\mathcal{E}_i$ is the product of the length density functions for
outgoing edges from $i$, 
and the vertex gadget $\mathcal{V}_i$ models branching at vertex $i$. 
In case a path does not include $i$, then $u=\bot \wedge v=\bot$ represents the case.
By our algorithm, 
we execute the convolutions across the all vertex and edge gadgets.
The integrand is a product of all of these gadgets. 
Considering only the product of the vertex gadgets yields a probability 
equivalent to the Iverson bracket evaluation across all conditions $\mathcal{V}_i(\Vec{z}_{E(i)})$ for $i=1, \dots, n$, 
combined via logical ANDs, 
the Conjunctive Normal Form (CNF) of $\mathcal{U}_{ui,iv}$.
If we convert the form to DNF, each term is a logical product of 
$\mathcal{U}_{ui,iv}$ for all vertex $i$,
where each of $\mathcal{U}_{ui,iv}$ corresponds 
that vertex $i$ chooses incoming edge $ui$ and outgoing edge $iv$. 
Such choices define a functional subgraph for each DNF term defined as follows.
\begin{definition}
Let $G$ be a directed graph.
A subgraph $G'$ of $G$ is the functional subgraph
of $G$ if every vertex in $G'$ has outdegree $1$.
We set $\pfs_i$ as the family of all functional subgraphs of $G_i$.
\end{definition}

The functional subgraphs yield the set of all $s-t$ paths. 
We filter these using the condition $\mathrm{Garbage}(u i, i v)$ 
to remove non-essential structures.
By invoking the recurrence relation in Definitions~\ref{definition:FGi} and \ref{definition:FGi+1}, 
we obtain the comprehensive collection of edge sets $\pi \in \pfs_n$. 
The functional subgraph edge set $\pi$ is 
described as a DNF term by the existence of 
$\mathcal{U}_{ui,iv}$ after expanding the logical AND operations with other vertex gadgets. 
The condition restricts consideration to paths where all 
internal vertices satisfy both indegree and outdegree exactly one; 
specifically, every edge must be chosen by both its source and sink vertices.
The DNF and the derived family of functional graphs are necessary only for the proof; 
they are not loaded on the memory during algorithm execution.

Although $\mathrm{Garbage}(ui,iv)$ does not eliminate cyclic paths, 
the probability measure of any cycle is zero. 
Remember that each component of $\Vec{z}$ is the shortest path length
from an edge source or edge destination to the overall sink of graph.
Given the constraint 
$z_{e_1}\ge z_{e_2}\ge\cdots\ge z_{e_{|\pi|}}$ and then 
$z_{e_{|\pi|}}\le z_{e_1}$ 
for edges $e_1,\dots,e_{|\pi|}\in \pi$, 
assuming that all lengths are non-negative or non-positive, 
the cycle requires $z_{e_1}=\cdots =z_{e_{|\pi|}}$, 
which ends up with measure zero.

The running time of our algorithm on graphs with bounded treewidth 
is established by analyzing the complexity of calculating the piecewise polynomial 
function $F_{\partial G_i}(\Vec{z}_{\partial_E[i]})$. 
Expanding this function into the sum of products shows that 
the variation number of these terms can be rigorously bounded 
by a function of the underlying treewidth, since
the factors of piecewise polynomial's each term is limited to
the power of components of $\Vec{z}$, or step function $H$
with its input at most two components of $\Vec{z}$ and an integer
in $0,1,\dots,n$. This leads to the proof of Theorem~\ref{th:XPalgorithm}.

A comprehensive formalization and rigorous proof of the concepts 
discussed herein are presented in the following sections.

\section{Preliminaries}
We start with notations for vectors, vectors indexed by sets, and 
concatenation of two vectors.
\begin{definition}
\noindent {\bf Notation for Vectors.}
We denote by $\Vec{0}$ and $\Vec{1}$ the vector of all zeros and the vector of all ones, 
respectively.  
The dimensions of $\Vec{0}$ and $\Vec{1}$ 
are given by the context.
Let $S$ be a set.
For $N\subseteq \mathbb{R}$,
we write $\Vec{v}\in N^{S}$ so that $\Vec{v}$ is an $|S|$-dimensional
vector indexed by $S$; the components of $\Vec{v}$ are specified by $s\in S$,
where $\Vec{v}(s)=v_s\in N$ is a component indexed by $s\in S$. 
Let $S'\subseteq S$. Then $\Vec{u}=\Vec{v}_{S'}\in N^{S'}$ is a vector where 
$u_s=v_s$ for $s\in S'$.
Let $U$ and $V$ be two sets.
Let $\Vec{u}$ and $\Vec{v}$ be two vectors.
We denote the concatenation by $\Vec{w}=\cat{\Vec{u}}{\Vec{v}}$ 
a vector indexed by $U\cup V$,
assuming that $U\cap V=\emptyset$.
We set the components $w_q=u_q$ if $q\in U$
and $w_q=v_q$ if $q\in V$.
\end{definition}

The following are the definitons of the basic terms for probability.
\begin{definition}
\noindent {\bf Notation in Probability.}
Let $X$ be a random variable.
The distribution function of $X$ is $\Pr[X\le x]$.
The density function of $X$ is $\frac{d}{dx} \Pr[X\le x]$.
It is sometimes convenient to consider the complementary 
distribution function $1-\Pr[X\le x]$; in case the distribution
function of $X$ is denoted by $F(x)$, the complementary distribution
function is denoted by $\bar F(x)=1-F(x)$.
Let $X_1,\dots,X_m$ be a sequence of random variables.
If $X_1,\dots,X_m$ share the same distribution function, the random variables
are identically distributed.
A sequence of random variables 
$X_1,\dots, X_m$ is called i.i.d. (identically and independently distributed),
if they are mutually independent and share a common distribution function.
We use the notation $X_e\iidsim F$ for $e\in E$, to signify that 
$\Vec{X}\in \mathbb{R}^E$ is a random vector where each component
$X_e$ is mutually independent and has distribution function $\Pr[X_e\le x]=F(x)$.
\end{definition}

We define some specific distributions that appear repeatedly throughout
this paper. 
\begin{definition}
\noindent{\bf Specific Distributions.}
Let $H(x)$ be the Heaviside step function ($H(x)=0$ for $x<0$, otherwise $H(x)=1$).
If $X$ has density function $\frac{d}{dx}\Pr[X\le x]=H(x)H(1-x)$, 
we say that $X$ is uniformly distributed over $[0,1]$.
Here, we set $\Unif(x)=x H(x)H(1-x) + H(x-1)$ and $\unif(x)=H(x)H(1-x)$  
to be the uniform distribution function and its density function, respectively.
Here, the distribution and the density functions 
of $m$ uniform random variables i.i.d. over $[0,1]$ are
denoted by $\Unif_m(x)$ and $\unif_m(x)$, respectively.
For a static value $c$, the distribution function is $\Pr[c\le x]=H(x-c)$,
whose derivative is the Dirac delta $\delta(x-c)$.
\end{definition}
To maintain the analytical consistency, we consider $\delta(x)$ to be the
generalized derivative of $H(x)$, which can be viewed as a narrowed uniform distribution.
One important property of $\delta(x)$ is 
$\int_{\mathbb{R}}F(x-t)\delta(t-c)dt = F(x-c)$ (sifting property),
implying that the convolution with $F(x)$ and $\delta(x-c)$  
ends up with a translation of $F(x)$.

We formalize our problems through the following two definitions.
\begin{definition}
\noindent{\bf Graphs and Paths} 
Let $G=(V,E)$ be a directed graph with vertex set $V=\{1,\dots,n\}$ 
and edge set $E\subseteq V\times V$.
For each $e=(u,v) \in E$ (denoted by $uv$ for short), 
its length is a random variable $X_{uv}$ with 
distribution function $F_e(x)=\Pr[X_e\le x]$.
We denote by $\Pi_G$ the set of all simple paths in $G$. 
\end{definition}

\begin{definition}
\noindent{\bf Shortest and Longest Path Length Distribution Function (SPPDF and LPPDF):} 
Given a graph $G=(V,E)$ and a random vector $\Vec{X}\in \mathbb{R}^E$,
compute the probabilities $F_{\rm MIN}^{\Vec{X}}(x)$ and 
$F_{\rm MAX}^{\Vec{X}}(x)$, where 
\begin{align*}
    F_{\rm MIN}^{\Vec{X}}(x)&=\Pr\left[\left(\min_{\pi\in \Pi_G}\sum_{e\in \pi} X_e\right) \le x\right]=\Pr\left[\bigvee_{\pi\in \Pi_G}\left(\sum_{e\in \pi} X_e \le x\right)\right],\\
    F_{\rm MAX}^{\Vec{X}}(x)&=\Pr\left[\left(\max_{\pi\in \Pi_G}\sum_{e\in \pi} X_e\right) \le x\right]=\Pr\left[\bigwedge_{\pi\in \Pi_G}\left(\sum_{e\in \pi} X_e \le x\right)\right].
\end{align*}
If the edge lengths are i.i.d., we call the problems SPPDF-IID and LPPDF-IID, respectively.
We sometimes consider the complementary distribution function 
$\bar F_{\rm MIN}^{\Vec{X}}(x)=1-F_{\rm MIN}^{\Vec{X}}(x)$ for the convenience of description.
\end{definition}

Notice that the problem of finding out ``the shortest path'' or 
``the longest path'' is another separate problem (see, e.g., \cite{NKBM2006}).
If the random edge length vector $\Vec{X}$ is 
clear from context, we omit $\Vec{X}$ and write
$F_{\rm MIN}(x)$ and $F_{\rm MAX}(x)$.
For undirected graphs, we regard the undirected edge as a anti-parallel 
directed edge pair, considering only simple paths.

We review here the definitions of tree decomposition
and treewidth. See, e.g.,~\cite{Kloks1994}
for more on this topic.
Since we consider directed graphs in this paper, 
we import the treewidth for undirected graphs to our problem.
\begin{definition}
\noindent{\bf Underlying Tree Decomposition and Underlying Treewidth}:
  Given a directed graph $G=(V, E)$. Then,
  the underlying tree decomposition of $G$ is a pair
  ${\cal T}=({\cal B},{\cal A})$ such that
  ${\cal T}$ is a tree satisfying the following conditions.\\
  \noindent (TD1). $\bigcup_{B\in{\cal B}}B=V$;\\
  \noindent (TD2). $uv\in E\Rightarrow \exists B\in {\cal B}~\text{s.t.}~uv\subseteq B$;\\
  \noindent (TD3). for all $B_h,B_i,B_j \in {\cal B}$, if $B_i$ is on a path
  between $B_h$ and $B_j$ in {\cal T}, then $B_h\cap B_j\subseteq B_i$.
  
  We call $B\in {\cal B}$ a bag.
  The tree decomposition ${\cal T}$ has width $k$
  if $|B|\le k+1$ for all $B\in{\cal B}$.
  The graph $G$ has treewidth $k$ if the minimum width of
  all tree decompositions of $G$ is at most $k$.
  The underlying treewidth of $G$ is the treewidth of $G$'s underlying undirected graph.
\end{definition}

\section{$\sharpp$-hardness}
In this section, we prove that SPPDF-IID and LPPDF-IID are $\sharpp$-hard problems
when edge lengths are i.i.d. uniformly distributed.
We first establish the $\sharpp$-hardnes of SPPDF and LPPDF utilizing both
static zero length edges and random length edges i.i.d. distributed uniformly over
$[0,1]$.
We then prove the $\sharpp$-hardness of SPPDF-IID and
LPPDF-IID for a DAG in case edge lengths are distributed over $[0,1]$.
By extending the result, we prove that 
SPPDF-IID and LPPDF-IID are $\sharpp$-hard also for an undirected graph
in case edge lengths are i.i.d. and obey the uniform distribution
over $[1,2]$ or $[-2,-1]$.
We then extend the result to the other probability distributions
by the local uniformity.

\subsection{$\sharpp$-hardness using zero-length edges}
Here, we show that any instance of $\#$LE2H can be reduced to 
SPPDF or LPPDF, incorporating both fixed zero length edges and
 random edge lengths i.i.d. uniformly over $[0,1]$.
As mentioned, before our $\sharpp$-hardness results build upon 
the hardness of the following problem, due to \cite{DP2020}.
\definitionLEtwoH*
\thDP*
We show that any instance of $\#$LE2H can be reduced to SPPDF and LPPDF 
for the following DAG whose example is shown in 
Fig.~\ref{fig:reduction}~(left). 
\begin{definition}
\label{definition:graphzero}
Given a poset $\mathcal{P}=(U\cup V, R)$,
we set
\begin{align*}
\graphzero_{R}\defeq (U\cup V\cup \{s,t\}, R \cup P)~~\text{where}~~P=P_U\cup P_V=\{su|u\in U\}\cup \{vt| v\in V\}.
\end{align*}
We call $\graphzero_R$ the {\em transportation graph} related to poset $\mathcal{P}$.
\end{definition}

The number of linear extensions ($\#$LE) bridges between
the discrete counting problems and the continuous polytope volumes, 
namely the order polytope.
The following is the special case of Stanley's work \cite{Stanley1986} 
restricting the poset's height up to two.
\thorderpolytope*

We consider the case where some edges
has static $0$ length and the other 
has random edge lengths i.i.d. uniformly over $[0,1]$.
We denote the vector of edge lengths by $\cat{\Vec{0}}{\Vec{X}}$ using 
a random vector $\Vec{X}\in\mathbb{R}^P$, meaning that the length is $X_e$
for $e\in P$ and $0$ for the other edges.
As mentioned before,
SPPDF and LPPDF are $\sharpp$-hard 
by the reduction from $\#$LE2H as follows.
\lemmaUniformhard*
We note that a one-to-one correspondence does not exist 
between an arbitrary linear extension of $\mathcal{P}$ 
and an $s-t$ path within the graph $\graphzero_R$. 
Instead, the actual bijection exists between each linear extension 
and the order polytope associated with a directed path whose 
vertices are entirely comparable to each other.
A clearer understanding of this correspondence can be achieved 
by considering the realized values of random edge lengths 
drawn from the uniform distribution. 
If we fix any vector $\Vec{X} \in [0, 1]^P$, 
the corresponding linear extension is given by 
verifying whether every edge length $uv$ satisfies the condition $X_{su}+X_{vt}\le 1$. 
This condition is equivalent to testing $X_{su}\le Y_{vt}$ by defining $Y_{vt}=1-X_{vt}$, 
where $Y_{vt}$ is also i.i.d. uniformly distributed. 
Since the linear extension is the order of the magnitude of the components 
of $\cat{\Vec{X}_U}{\Vec{Y}_V}$,
the linear extension is considered valid if all these inequalities are satisfied, 
implying that $\Vec{X}$ belongs to the order polytope.
\begin{proof} (of Lemma~\ref{lemma:Uniform-hard})
By Theorem~\ref{th:orderpolytope}, we would solve $\#$LE2H
if we could compute the probability 
\begin{align}
Q=\Pr\left[\bigwedge_{uv\in R} (X_u\le X_v)\right]=\frac{\#\text{LE2H}}{|P|!}.
\end{align}

If we could compute $F_{\rm MAX}^{\cat{\Vec{0}}{\Vec{X}}}(1)=\bar F_{\rm MIN}^{\cat{\Vec{0}}{-\Vec{X}}}(-x)$,
we would have 
\begin{align}
F_{\rm MAX}^{\cat{\Vec{0}}{\Vec{X}}}(1)=\Pr\left[\bigwedge_{\pi\in \Pi_G(1,n)}\left(\sum_{e\in \pi} X_e \le 1\right)\right]=\Pr\left[\bigwedge_{uv\in R}(X_{su}+X_{vt}\le 1)\right].
\end{align}
Therefore, for each edge $uv\in R$, 
there exists exactly one condition $X_{su}+X_{vt}\le 1$  
that is equivalent to $X_{su}\le 1-X_{vt}$. 
Notice that, since $X_{vt}$ is a random variable 
uniformly distributed over $[0,1]$ independent from the other edge lengths, 
so is $Y_{vt}=1-X_{vt}$ by the definition of the uniform distribution. 
By using the new edge length $Y_{vt}$, 
we can see $Q=F_{\rm MAX}^{\cat{\Vec{0}}{\Vec{X}}}(1)$ since 
\begin{align}
F_{\rm MAX}^{\cat{\Vec{0}}{\Vec{X}}}(1)=\Pr\left[\bigwedge_{uv\in R}(X_{su}+X_{vt} \le 1)\right]=\Pr\left[\bigwedge_{uv\in R}(X_{su}\le Y_{vt})\right]=Q.
\end{align}

Since
$F_{\rm MAX}^{\cat{\Vec{0}}{\Vec{X}}}(1)=\bar F_{\rm MIN}^{\cat{\Vec{0}}{-\Vec{X}}}(-1)=\bar F_{\rm MIN}^{\cat{\Vec{0}}{\Vec{1}-\Vec{X}}}(2-1)=\bar F_{\rm MIN}^{\cat{\Vec{0}}{\Vec{X}}}(1)$ 
by Proposition~\ref{proposition:minmaxsymmetry} and 
the definition of the transportation graph, 
we have the $\sharpp$-hardness of computing 
$\bar F_{\rm MIN}^{\cat{\Vec{0}}{\Vec{X}}}(1)$.
\end{proof}

\subsection{$\sharpp$-hardness in case zero-length edges are not allowed}
Here, we prove $\sharpp$-hardness for the case where all edge 
lengths are i.i.d. uniformly over $[0,1]$. 
We subdivide edges in $R$ of the transportation graph.
Then, using Theorem~\ref{th:DAGLP}, 
we transform the computation of the longest path length 
distribution function into 
a multivarate convolution of some functions 
$\Psi_U(x,\Vec{y})$ and $\Psi_V^+(\Vec{y})$, or 
$\Psi_U(x,\Vec{y})$ and $\tilde \Psi_V(\Vec{y})$ 
for $\Vec{y}\in \mathbb{R}^{R}$.
The transformation gives us a clearer foresight because 
$\Psi_U(x,\Vec{y}), \Psi_V^+(\Vec{y})$,  
and $\tilde \Psi_V(\Vec{y})$ 
can be explicitly computed.

\begin{definition}
Let $\mathcal{P}=(U\cup V,R)$ and 
$\graphzero_R=(U\cup V\cup\{s,t\},R\cup P)$ be a height-2 poset and its related
transportation graph.
We have $\graphtwo_R$ by subdividing all edges in $R$ of $\graphzero_R$.
We obtain $\graphone_R$ by 
subdividing all edges in $R$ of $\graphzero_R$, 
and by adding a path of $|R|$ edges at the sink $t$,
where the new sink is denoted by $t'$.
\end{definition}
Let $\graphthree_R$ be the transportation graph 
related to $\mathcal{P}=(U\cup V\cup\{s,t\}, R\cup P)$ 
with all edge lengths i.i.d. uniformly over $[0,1]$.
We show an example of $\graphtwo_R$ and 
$\graphone_R$ in Figure~\ref{fig:subdivision} (b) and (d).
Clearly, the longest path length distribution of $\graphthree_R$ 
(Fig.~\ref{fig:subdivision} (a))
coincides with that of a graph $\graphtwo_R$
(Fig.~\ref{fig:subdivision} (b)).

\begin{figure}[ht]
\begin{center}
\includegraphics[clip,scale=0.7, bb=10 620 600 850]{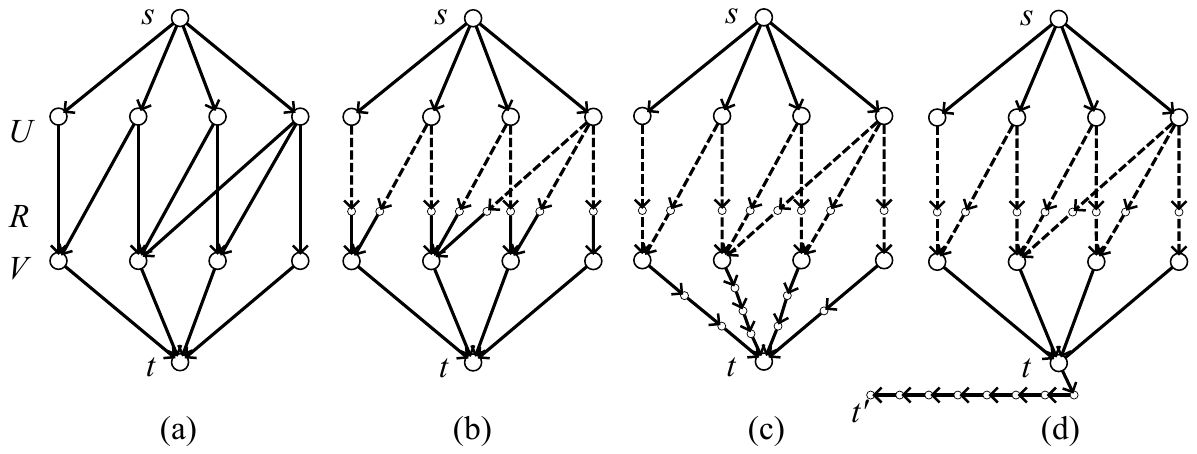}
\end{center}
\vspace*{-10pt}
\caption{Example of subdivision graphs 
(a) $\graphthree_R$, (b) $\graphtwo_R$ and (d) $\graphone_R$. 
(c) is an intermediate that appears in the transformation.
Lemma~\ref{lemma:equalPsi} proves the longest path length distribution function of (d) $\graphone_R$ 
is proportional to that of (a) $\graphthree_R$ that is with completely i.i.d. edge lengths.}
\label{fig:subdivision}
\end{figure}

The two graphs $\graphtwo_R$ and $\graphone_R$ 
above are essential to our proof.
In $\graphtwo_R$, we set the length of the edge between $U$ and $R$ 
as a static value $0$; any other edge $e$ in $\graphtwo_R$ has edge length $X_e\iidsim\Unif$.
In $\graphone_R$, we set the length of the edge incident to 
any one of $R$ as a static value $0$.
Any other edge $e$ in $\graphone_R$ has edge length $X_e\iidsim \Unif$.
Let $\Psi_{\rm MAX}(x)$ be the longest path length distribution function
of $\graphzero_R=(U\cup V\cup \{s,t\}, R\cup P)$ 
with $0$ length for edges in $R$ and i.i.d. uniformly
over $[0,1]$ random length for edges in $P$.
We state that we can compute $\Psi_{\rm MAX}(x)$ 
if we have the value of $\Psi_{\rm MAX}^+(x)$.
\begin{proposition}
\label{proposition:tail}
Let $\Psi_{\rm MAX}^+(x)$ be the longest path length distribution function 
from $s$ to $t'$ in $\graphone_R$.
For $0\le x\le 1$ and a constant $A$, we have
\begin{align*}
\Psi_{\rm MAX}^+(x)=Ax^{|P|+|R|}=\frac{|P|!}{(|P|+|R|)!}x^{|R|} \Psi_{\rm MAX}(x). 
\end{align*}
\end{proposition}
\begin{proof}
Remember that $\Psi_{\rm MAX}(x)=Ax^{|P|}$ by 
Proposition~\ref{proposition:constant}, and that 
the difference of $\graphzero_R$ and $\graphone_R$ is 
whether or not the graph has $t-t'$ path.
Then, since $t-t'$ path is always in the $s-t'$ path, 
we have
\begin{align*}
\Psi_{\rm MAX}^+(x)&=\Pr[X_{\rm MAX}+X_t\le x]=
\int_0^x\Unif_{|R|}(x-s)\frac{d}{dx}\Psi_{\rm MAX}(s) ds \\
&=\int_0^x\frac{(x-s)^{|R|}}{|R|!}A |P|s^{|P|-1} ds,
\end{align*}
where $X_{\rm MAX}$ is the longest path length in $\graphzero_R$, 
and $X_t$ is the length of $t-t'$ path.
This is computed by repeating the integration by parts, which gives
\begin{align*}
&\Psi_{\rm MAX}^+(x)=\left[\frac{(x-s)^{|R|-1}}{(|R|-1!}\frac{s^{|P|}}{|P|}\right]_0^x +\int_0^x\frac{(x-s)^{|R|-1}}{(|R|-1)!}A \frac{s^{|P|}}{|P|} ds=\int_0^x\frac{(x-s)^{|R|-1}}{(|R|-1)!}A \frac{s^{|P|}}{|P|} ds \\
&=\int_0^x\frac{(x-s)^{|R|-2}}{(|R|-2)!}A \frac{s^{|P|+1}}{|P|(|P|+1)} ds=\cdots=\int_0^x A\frac{s^{|P|+|R|-1}}{|P|(|P|+1)\cdots(|P|+|R|-1)} ds\\
&=A x^{|P|+|R|}\frac{|P|!}{(|P|+|R|)!}=\frac{|P|!}{(|P|+|R|)!}x^{|R|}\Psi_{\rm MAX}(x).
\end{align*}
\end{proof}

To our surprise,
we can shift the random length of the edges in $R$ to the ``tail'' of $\graphone_R$.
Let $\tilde \Psi_{\rm MAX}(x)$ be the longest path length 
distribution function of $\graphtwo_R$.
In computing $\tilde \Psi_{\rm MAX}(x)$ and $\Psi_{\rm MAX}^+(x)$, 
we separate the integrand of Theorem~\ref{th:DAGLP} into $\Psi_U(x,\Vec{y})$, 
$\tilde \Psi_V(\Vec{y})$ and $\Psi_V^+(\Vec{y})$ for $\Vec{y}\in\mathbb{R}^R$ 
so that:
\begin{restatable}{proposition}{propositionPsiMax} \label{proposition:PsiMax}
There exist $(|R|+1)$-variable function $\Psi_U(x,\Vec{y})$ and 
$|R|$-variable functions $\tilde \Psi_V(\Vec{y})$ and $\Psi_V^+(\Vec{y})$ 
for $\Vec{y}\in \mathbb{R}^R$ satisfying
\begin{align*}
\tilde \Psi_{\rm MAX}(x)&=\int_{\mathbb{R}^R}\Psi_U(x,\Vec{y}) \tilde\psi_V(\Vec{y})d\Vec{y}, ~~\text{and}\\
\Psi_{\rm MAX}^+(x)&=\int_{\mathbb{R}^R} \Psi_U(x,\Vec{y}) \psi_V^+(\Vec{y})d\Vec{y}, 
\end{align*}
where $\psi_V^+(\Vec{y})=D_R\Psi_V^+(\Vec{y})$ and
$\tilde \psi_V(\Vec{y})=D_R\tilde\Psi_V(\Vec{y})~~ (D_R=\prod_{uv\in R}\frac{\partial}{\partial y_{uv}})$.
That is,
\begin{align}
\Psi_U(x,\Vec{y})&=\int_{\mathbb{R}^U}\prod_{u\in U}\Unif(x-z_u)\frac{\partial}{\partial z_u}\prod_{v\in V_u}H(z_u-y_{uv})d\Vec{z}, \label{form:Psi_U0} \\
\tilde \Psi_V(\Vec{y})&=\int_{\mathbb{R}^V}\prod_{uv\in R} \Unif(y_{uv}- z_v)\prod_{v\in V} \unif(z_u) d\Vec{z}, ~~\text{and} \label{form:tildePsi_V0}\\
\Psi_V^+(\Vec{y})&=\int_{\mathbb{R}^V}\prod_{uv\in R} H(y_{uv}- z_v)\int_{\mathbb{R}}\prod_{v\in V} \unif(z_u-t) \unif_{|R|}(t)dt d\Vec{z}. \label{form:Psiplus_V0}
\end{align}
\end{restatable}
\begin{proof}
The claim is immediate by Theorem~\ref{th:DAGLP}\footnote{
Though Proposition~\ref{proposition:PsiMax} is immediate by 
Theorem~\ref{th:DAGLP}, we show a longer version of the proof
in Section~\ref{section:supplementary} , which is a demonstration of Theorem~\ref{th:DAGLP}
for the specific graphs, $\graphtwo_R$ and $\graphone_R$.
}.
\end{proof}

The following lemma shows that $\tilde \Psi_{\rm MAX}(x)$ is 
proportional to $\Psi_{\rm MAX}^+(x)$.
\begin{lemma}
\label{lemma:equalPsi}
For $x\le 1$, we have 
\begin{align*}
\tilde \Psi_{\rm MAX}(x)=|R|!\Psi_{\rm MAX}^+(x).
\end{align*}
\end{lemma}
\begin{proof}
By Proposition~\ref{proposition:psi_V} that is shown later, we have that 
\begin{align*}
\tilde \psi_V(\Vec{y})=|R|!\psi_V^+(\Vec{y}) 
\end{align*}
for $\Vec{y}\le \Vec{1}$.
Observe that, by Proposition~\ref{proposition:PsiMax},
both integrands of $\tilde \Psi_{\rm MAX}(x)$
and $\Psi_{\rm MAX}^+(x)$ are zero
when any one of the dummy variables is greater than $x$.
This is because the integrand of $\Psi_U(x,\Vec{y})$
has factors $\Unif(x-z_u)$ and $H(z_u-y_{uv})$ for $u\in U$ and
$uv\in R$.
Therefore, the coefficient $|R|!$ is multiplied to 
the entire integral throughout $\mathbb{R}^V$, the interval 
of integral.
\end{proof}
By the above arguments, we establish Theorem~\ref{th:directedhard} as follows.
\begin{proof}(of Theorem~\ref{th:directedhard})
By Lemma~\ref{lemma:equalPsi} and Proposition~\ref{proposition:tail},
we have 
\begin{align*}
\Psi_{\rm MAX}(x)=x^{-|R|}\frac{(|P|+|R|)!}{|P|! |R|!} \tilde \Psi_{\rm MAX}(x).
\end{align*}
If we could compute the longest path length distribution function
$\tilde\Psi_{\rm MAX}(x)$ of $\graphtwo_R$, 
we would have $\Psi_{\rm MAX}(x)$, which
is $\sharpp$-hard by Lemma~\ref{lemma:Uniform-hard}.
\end{proof}

To complete the proof of Lemma~\ref{lemma:equalPsi},
we obtain the following by direct calculation
\begin{restatable}{proposition}{propositionPsiU}
\label{proposition:psi_U}
Let $V_u\subseteq V$ be the set of children of $u\in U$.
We have
\begin{align*}
&\Psi_U(x,\Vec{y})
=\prod_{u\in U}\Unif\left(x-\max_{v\in V_u}\{y_{uv}\}\right). 
\end{align*}
\end{restatable}
\begin{proof}
Remember that 
$H(x)$ and 
$\frac{d}{dx}H(x)=\delta(x)$, the Dirac delta, 
is the distribution function and the density function of
static value $0$, respectively.
By Proposition~\ref{proposition:PsiMax}, we have
\begin{align*}
\Psi_U(\Vec{x},\Vec{y})
&=\int_{\mathbb{R}^U}\prod_{u\in U} \left(\Unif(x-z_u) \frac{\partial}{\partial z_u} \prod_{v\in V_u} H(z_u-y_{uv})\right) d\Vec{z}.
\end{align*}
Since $\prod_{v\in V_u}H(z_u-y_{uv})$ is nonzero if and only if 
$z_u-y_{uv}\ge 0$ for all $uv\in R$, we have $\prod_{v\in V_u}H(z_u-y_{uv})=H(z_u-\max_{v\in V_u}\{y_{uv}\})$. Therefore, the above is equal to 
\begin{align*}
&\int_{\mathbb{R}^U}\prod_{u\in U} \Unif(x-z_u) \frac{\partial}{\partial z_u} H\left(z_u-\max_{v\in V_u}\{y_{uv}\}\right) d\Vec{z}\\
&=\int_{\mathbb{R}^U}\prod_{u\in U} \Unif(x-z_u)\>\> \delta\left(z_u-\max_{v\in V_u}\{y_{uv}\}\right) d\Vec{z},
\end{align*}
which proves the claim since we have
$\int_{\mathbb{R}} \Unif(x-z)\delta(z-c)dz = \Unif(x-c)$ 
by the sifting property of Dirac delta.
\end{proof}

\begin{restatable}{proposition}{PropositionPsiV}
\label{proposition:psi_V}
For $\Vec{y}\le \Vec{1}$, we have
\begin{align*}
\tilde \Psi_V(\Vec{y})=|R|!\Psi_V^+(\Vec{y}).
\end{align*}
\end{restatable}
\begin{proof}
By Proposition~\ref{proposition:PsiMax}, we have 
\begin{align*}
\tilde \Psi_V(\Vec{y})
&=\int_{\mathbb{R}^V}\left(\prod_{uv\in R}\Unif(y_{uv}-z_v)\right)\prod_{v\in V}\unif(z_v) d\Vec{z}\\
&=\int_{\mathbb{R}^V}\left(\prod_{uv\in R}(y_{uv}-z_v)H(y_{uv}-z_v)H(1-(y_{uv}-z_v))\right)\prod_{v\in V}\unif(z_v) d\Vec{z}.
\end{align*}
Since there are factors $H(y_{uv}-z_v)$ and $\unif(z_v)$ for every $uv\in R$ and $\Vec{y}\le \Vec{1}$ by assumption, the integrand is nonzero
only if $z_v\in [0,1]$ for each $v\in V$.
The above is equal to 
\begin{align*}
&\int_{[0,1]^V}\left(\prod_{uv\in R}(y_{uv}-z_v)H(y_{uv}-z_v)H(1-(y_{uv}-z_v))\right)\prod_{v\in V}\unif(z_v) d\Vec{z}.
\end{align*}
Since $H(1-(y_{uv}-z_v))=1$ for $0\le z_v\le y_{uv}\le 1$, this is
equal to
\begin{align*}
&\int_{[0,1]^V}\left(\prod_{uv\in R}(y_{uv}-z_v)H(y_{uv}-z_v)\right)\prod_{v\in V}\unif(z_v) d\Vec{z}.
\end{align*}
By setting $U_v$ as the set of parents of $v\in V$, 
we can transform the above as 
\begin{align*}
&\int_{[0,1]^V}\left(\prod_{v\in V}\prod_{u\in U_v}(y_{uv}-z_v)H(y_{uv}-z_v)\right)\prod_{v\in V}\unif(z_v) d\Vec{z}\\
&=\int_{[0,1]^V}\left(\prod_{v\in V}Q_v(\Vec{y}_{U_v},z_v)\right)\prod_{v\in V}\unif(z_v) d\Vec{z},
\end{align*}
where $\Vec{y}_{U_v}\subseteq [0,1]^{U_v}$ is a subvector of $\Vec{y}$ and $Q_v(\Vec{y}_{U_v},z_v)=\prod_{u\in U_v}(y_{uv}-z_v)H(y_{uv}-z_v)$.
Then, by the commutativity of convolution with respect to $z_v$'s for $v\in V$, 
we consider the $|U_v|$-th derivative of $Q_v(\Vec{y}_{U_v},z_v)$ and
$|U_v|$-th antiderivative of $\unif(z_v)$ with respect to $z_v$.
Remember that $\unif_m(x)$ is the density of $m$ i.i.d. uniformly random variables 
distributed over $[0,1]$.
Also, the $m$-th antiderivative of $\unif(x)$ with respect to $x$ is 
$\unif_{m+1}(x)$ for $0\le x\le 1$. Now, we have 
\begin{align*}
\tilde \Psi_V(\Vec{y})=\int_{[0,1]^V}\left(\prod_{v\in V}\left(-\frac{\partial}{\partial z_v}\right)^{|U_v|}Q_v(\Vec{y}_{U_v},z_v)\right)\prod_{v\in V}\unif_{|U_v|+1}(z_v) d\Vec{z}.
\end{align*}
Since, by definition, we have 
\begin{align*}
\left(-\frac{\partial}{\partial z_v}\right)^{\!\!|U_v|}Q_v(\Vec{y}_{U_v},z_v)=\left(-\frac{\partial}{\partial z_v}\right)^{\!\!|U_v|}\!\prod_{u\in U_v}\!(y_{uv}\!-\!z_v)H(y_{uv}\!-\!z_v)=|U_v|!\!\prod_{u\in U_v}\!\!H(y_{uv}\!-\!z_v),
\end{align*}
the above arguments results in 
\begin{align*}
\tilde \Psi_V(\Vec{y})&=\left(\prod_{v\in V}|U_v|!\right)\int_{[0,1]^V}\left(\prod_{u\in U_v}H(y_{uv}-z_v)\right)\prod_{v\in V}\unif_{|U_v|+1}(z_v) d\Vec{z}.
\end{align*}
Notice that this form is equal to 
\begin{align*}
\left(\prod_{v\in V}|U_v|!\right)\Pr\left[\bigwedge_{uv\in R}\left(\sum_{uv\in R}X_{uv} \le y_{uv}\right)\right].
\end{align*}
This is as though we consider yet another graph $G_R$ by replacing
each edge $vt\in P_V$  of $G_R$ by a path $\pi_{vt}$ of $|U_v|+1$ 
edges (illustrated in Fig.~\ref{fig:subdivision} (c)).
In other words, we removed the random length from each of the incoming 
edges of $v$ and gave the lengths to $v$'s outgoing edges.

Consider connecting length zero edge at the sink, where the density function 
of the length is the Dirac delta.
We have
\begin{align*}
\tilde \Psi_V(\Vec{y})&=\left(\prod_{v\in V}|U_v|!\right)\int_{[0,1]^V}\int_{\mathbb{R}}\left(\prod_{uv\in R}H(y_{uv}-z_v)\right)\left(\prod_{v\in V}\unif_{|U_v|+1}(z_v-t)\right)\delta(t) dt d\Vec{z}.
\end{align*}
Then, we once again shoehorn the length of edges between $V$ and $t$ 
to the ``tail'' edge by the commutativity of convolution 
with respect to $t$. We consider the $\sum_{v\in V}|U_v|$-th derivative of 
$\prod_{v\in V}\unif_{|U_v|+1}(z_v-t)$ and 
the $\sum_{v\in V}|U_v|$-th antiderivative
of $\delta(t)$ with respect to $t$.
Notice that $\sum_{v\in V}|U_v|=|R|$ and $\delta(t)=\unif_0(t)$ for
$0\le t \le 1$.
We have
\begin{align*}
&\tilde \Psi_V(\Vec{y})=\left(\prod_{v\in V}|U_v|!\right)\int_{[0,1]^V}\int_{\mathbb{R}}\left(\prod_{uv\in R}H(y_{uv}-z_v)\right)\left(\prod_{v\in V}\unif_{|U_v|+1}(z_v-t)\right)\delta(t) dt d\Vec{z}.
\end{align*}
Here, we apply the commutativity of convolution with respect to $t$,
meaning that the above is equal to 
\begin{align*}
&\left(\prod_{v\in V}|U_v|!\right)\int_{[0,1]^V}\int_{\mathbb{R}}\left(\prod_{uv\in R}H(y_{uv}-z_v)\right)\frac{|R|!}{\prod_{v\in V}|U_v|!}\left(\prod_{v\in V}\unif(z_v-t)\right)\unif_{|R|}(t) dt d\Vec{z}\\
&=|R|!\int_{[0,1]^V}\int_{\mathbb{R}}\left(\prod_{uv\in R}H(y_{uv}-z_v)\right)\left(\prod_{v\in V}\unif(z_v-t)\right)\unif_{|R|}(t) dt d\Vec{z}
=|R|!\Psi_V^+(\Vec{y}),
\end{align*}
where $X_t$ is the sum of 
$|R|$ i.i.d. uniformly over $[0,1]$ random variables,
proving the claim.
\end{proof}

\subsection{$\sharpp$-hardness for undirected graphs}
We extend the result to undirected graphs.
In the following, we consider the underlying undirected graph 
of the transportation graph.
The difference between the undirected graph and the directed graph is
that there may be a path from $s$ to $t$ that goes from a vertex in $V$ to 
a vertex in $U$ on the way. 
We circumvent the obstacle
by considering uniform random lengths i.i.d. over $[1,2]$.
In what follows, we denote the uniform distribution over $[1,2]$ by $\Unif^+$;
$\Unif^+_{m}(x)$ and $\unif^+_{m}(x)$ are the sum distribution and density, 
respectively, of $m$ i.i.d. random variables drawn from $\Unif^+$.

\begin{restatable}{corollary}{corollaryUndirectedPositiveShortestNegativeLongest}
\label{corollary:undirected_positive_shortest_negative_longest}
Given an undirected graph $G=(V,E)$, SPPDF-IID (resp. LPPDF-IID) is
$\sharpp$-hard if the edge length are uniformly distributed over $[1,2]$ (resp. $[-2,-1]$).
\end{restatable}
\begin{proof}
Let us focus on SPPDF-IID. 
Given a height-2 poset $\mathcal{P}=(U\times V,R)$ and 
its related (directed) transportation graph $\graphzero_R=(U\cup V\cup \{s,t\},R\cup P)~~P=P_U\cup P_V$,
where $P_U=\{su|u\in U\}$, $P_V=\{vt| v\in V\}$, and each edge $e\in P$ has its length $X_e+1\iidsim\Unif^+$.
Let $\graphfour_R$ be the underlying undirected graph of $\graphzero_R$. 
Let $F_{\rm MIN}^{\Vec{X}}(x)$ and $\acute{F}_{\rm MIN}^{\Vec{X}}(x)$ be the distribution functions of
the shortest path length in $\graphzero_R$ and $\graphfour_R$
with edge lengths given by $\Vec{X}$, respectively.

Clearly, the shortest path has exactly three edges if 
the edge lengths are i.i.d. uniformly over $[1,2]$, 
implying that
$\acute{F}_{\rm MIN}^{\Vec{X}+\Vec{1}}(x+3)=F_{\rm MIN}^{\Vec{X}}(x)$
for $0\le x \le 1$.
Therefore, $\sharpp$-hardness of SPPDF-IID 
for the uniform distribution over $[1,2]$ has been proved.

The argument for LPPDF-IID is symmetry by assuming the edge lengths
i.i.d. uniformly over $[-2,-1]$.
\end{proof}

We note that the above argument is not applicable for SPPDF-IID with $\Vec{X}\le \Vec{0}$
and LPPDF-IID with $\Vec{X}\ge \Vec{0}$. This is because, 
the number of edges may not be the same between the shortest/longest
path candidates in $\graphfour_R$ 
if we consider the random edge lengths given by 
$\Vec{X}\in\mathbb{R}^{R\cup P}$ 
with components $X_e\in [1,2]^{R\cup P}$.
In that case, we need to consider many cases,
which is not efficient.

To prove Theorem~\ref{th:undirectedhard},
we avoid the problem by considering yet another graph.
\begin{definition}
Given a height-2 poset ${\cal P}=(U\cup V,R)$, 
we consider another graph $\graphfive_R=(\hat U\cup V\cup \{s,t\},\hat E)$,
where 
$\hat U=U\cup U^+\cup U^-$, and $U^\pm=\{u^\pm| u\in V\}$ 
taking the signs respectively.
We define the edge set by
$\hat E=P_{U^+}\cup P_{U^-}\cup P_V \cup \hat P\cup R$,
where
\begin{align*}
P_{U^\pm}&=\{\{s,u^\pm\}|u\in U\}\\
\hat P&=\{\{u_i^+,u_i\},\{u_i,u_i^-\},\{u_i^+,u_j^-\}| u_i,u_j\in U, j-1\equiv i\!\!\!\!\mod |U|\},
\end{align*}
taking the signs respectively.
\end{definition}
\label{definition:graphfive}
Fig.~\ref{fig:Undirected_Reduction} (left) shows an example of $\graphfive_R$.
We then focus on LPPDF-IID in $\graphfive_R$, where
all edge lengths are i.i.d. uniformly over $[1,2]$.
There, the longest path visits $u^+\in U^+$ immediately after $s$,
then, the path must visit all $\hat U$ vertices ending at $u$ after $u^-$;
or, the longest path visits $u^-\in U^-$ immediately after $s$ ending at $u$ after $u^+$.

\begin{figure}[ht]
\begin{center}
\includegraphics[clip,scale=0.7, bb=40 660 410 820]{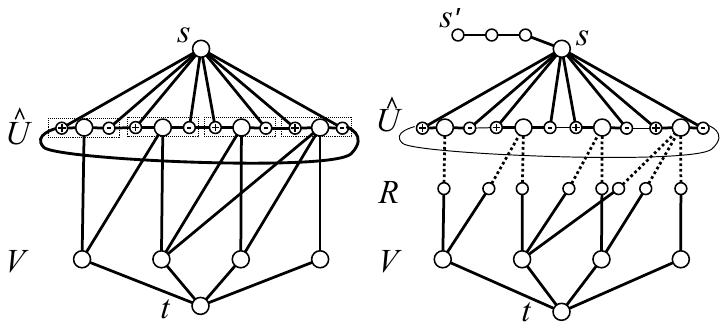}
\caption{Undirected graph $\graphfive_R$ for longest path problem reduction (left) and its transformation $\graphsix_R$(right). The dashed line edges have static length $0$. Thin line edges have static length $1$. Thick line edges have uniform random length i.i.d. over $[1,2]$.}
\label{fig:Undirected_Reduction}
\end{center}
\end{figure}

To deal with the cycle of $\graphfive_R$, we show the following
propositions about the longest path length in a cycle.
\begin{restatable}{proposition}{propositionRing}
\label{proposition:ring}
Let $U$ be the vertex set of $C$, whose edge set is $E$.
Consider a path that starts at $u\in U$ then go around the 
cycle $C=(U,E)$ in a ``clockwise'' directon ending at vertex
the ``anticlockwise''-neighbor $u'\in U$ of $u$.
We add dummy variable $z_u\in \mathbb{R}$ for
the length of the paths before $u$ and after $u'$.
Then, the longest path length distribution function is 
\begin{align*}
\Phi_C(x,\Vec{z})=\Pr\left[\bigwedge_{\{u,u'\}\in E}\left(\sum_{e\in E\setminus\{\{u,u'\}\}}X_e +z_u\le x\right)\right].
\end{align*}
In case the edge lengths are i.i.d. uniformly over $[0,1]$,
we have, for $0\le x \le 1$, 
\begin{align*}
\Phi_C(x,\Vec{z})=\mathrm{Vol}(K_C)\prod_{u\in U}(x-z_u)H(x-z_u),
\end{align*}
where
\begin{align*}
K_C=\left\{\Vec{x}\in [0,1]^{E}\middle|\bigwedge_{\{u,u'\}\in E}\left(\hspace*{-7mm}\sum_{\hspace*{7mm}e\in E \setminus\{\{u,u'\}\}}\hspace*{-7mm}x_e\le x\right)\right\}.
\end{align*}
\end{restatable}
\begin{proof}
Since the direct computation of $\Phi_C(x,\Vec{z})$ gets tedious, 
we consider the folllowing instead.
Firstly, $\Phi_C(x,\Vec{z})$ is the volume of an $|E|=|U|$-dimensional 
polytope that is totally contained by a hypercube $[0,x]^{E}$ 
for $0\le x \le 1$.
We characterize the polytope's volume $\Phi_C(x,\Vec{z})$ as follows.
For any $u\in U$, $\Phi_C(x,\Vec{z})$ has $(x-z_u)$ as its factor 
since $\Phi_C(x,\Vec{z})=0$ if there exists a $u$ such that $x-z_u=0$.
Also, $z_u$ for $u\in U$ appears in exactly one constraint, 
$\Phi_C(x,\Vec{z})$ is linear with respect to $z_u$.
Then, the constraint defining the polytope's size to each axis has $x$, 
$\Phi_C(x,\Vec{z})$ is a polynomial of $x$ with degree $|E|=|U|$.
In case $x=1$ and $\Vec{z}=\Vec{0}$, the volume is the constant 
coefficent, which is equal to the volume of $K_C$.
\end{proof}

Then, we coan compute the volume of $K_C$ as follows.
\begin{restatable}{proposition}{propositionKC}
\label{proposition:KC}
Polytope $K_C$ for a cycle $C=(U,E)$ has $|U|+2$ extreme points:
$\Vec{0}$, unit orthogonal base $\Vec{\mathrm{e}}_{e}~(e\in E)$,
and $\frac{1}{|U|-1}\Vec{1}$.
We have $\mathrm{Vol}(K_C)=\frac{1}{|U|!(|U|-1)}$,
so that $\Phi_C(x,\Vec{z})=\frac{1}{|U|!(|U|-1)}\prod_{u\in U}(x-z_u)H(x-z_u)$.
\end{restatable}
\begin{proof}
We first check $K_C$'s extreme points.
For contradiction, suppose that there exists an extreme point
$\Vec{p}\in K_C$ that is not enumerated in the above.
First, we clearly have $\Vec{p}\not\in\mathrm{Conv}(\{\Vec{0}\}\cup\{\Vec{\mathrm{e}}_e|e\in E\})$,
implying $\sum_{e\in E}p_e\ge 1$.
In intersection of the halfspace and $K_C$,
by considering 
the convex hull $\mathrm{Conv}(\{\Vec{\mathrm{e}}_e| e\in E\}\cup\{\frac{1}{|U|-1}\Vec{1}\})$
has exactly $|U|+1$ extreme points which consists of a simplex.
Then, for any $e'\in E$,
$\sum_{e\in E\setminus\{e'\}}x_e\le 1$ gives  
the facet-defining hyperplane of simplex 
$\mathrm{Conv}(\{\Vec{\mathrm{e}}_e|e\in E\setminus \{e'\}\}\cup\{\frac{1}{|U|-1}\Vec{1}\})$.
The above arguments lead to that any point $\Vec{p}$ satisfying all constraints 
are inside of the simplex, 
where $\Vec{p}$ can be given by the convex-combination of
the other extreme points, a contradiction. 
Therefore, the first claim holds.

Since $K_C$ is a polytope that is given by two simplices, 
the volume of $K_C$ is the sum of the volume of these two.
One simplex that includes $\Vec{0}$ has volume $1/|U|!$, 
the other simplex has bottom ($\mathrm{Conv}(\{\Vec{\mathrm{e}}_e|e\in E\})$) 
with area $\frac{\sqrt{|U|}}{(|U|-1)!}$;
the latter simplex has height
$\sqrt{\frac{1}{|U|(|U|-1)^2}}$, proving the claim.
\end{proof}

To apply Proposition~\ref{proposition:ring} and \ref{proposition:KC}, 
notice that some edges are always present in 
any $s-t$ path of $\graphfive_R$.
We transform $\graphfive_R$
so that these edges are connected even before $s$.
The graph $\graphsix_R$ that we obtain after the transformation is 
described as follows.
\begin{definition}
Let $\graphsix_R$ be a undirected graph with
vertex set $\{s,s',t\}\cup S\cup \hat U\cup R\cup V$.
The set $S$ is the vertices of the path,
which is connected before $s$ with $|U|$ edges with
random length i.i.d. uniformly over $[1,2]$; 
we define $s'$ as the new source.
Consider edges in cycle $\hat C$.
We set a random edge length drawn from i.i.d. $\Unif^+$,
for each edge except for 
$\{u_i,u_{i+1}\}~~(i=1,\dots,|U|-1)$ and $ \{u_{|U|},u_1\}$,
that have static length $1$. 
Also, we subdivde the edges in $R$ of $\graphfive_R$.
For each $r=uv\in R$, 
we have static length $0$ edge from $u\in U$ to $r$,
and also length $X_{uv}\iidsim \Unif^+$ edge from $r$ to $v\in V$.
\end{definition}
Fig.~\ref{fig:Undirected_Reduction} (right) shows an example of  
$\graphsix_R$ whose longest path length distribution
is that of $\graphfive_R$ translated by exactly $|U|$ to the right.
The point of the transformation is that we move the randomness in
these edges to the $s'-s$ path since
the longest path in $\graphfive_R$ 
always includes the edges
$\{u_i^+,u_{i+1}^-\}~(i=1,\dots,|\hat U|-1)$ and
$\{u_{|\hat U|}^+,u_1^-\}$.
Setting static length $1$ for these $|U|$ edges 
makes the longest path longer by exactly $|U|$,
which is necessary for preserving
the structure of the longest path through the transformation.

Now, we transform the longest path length distribution function 
of $\graphsix_R$ into a multivariate convolution of
two functions as in the previous subsection.
Here, dummy variables $z_u$ or $y_r$ represent
the longest path length from $u\in U$ or $r\in R$ to $t$,
respectively.
\begin{proposition}
\label{proposition:Phi_Psi_convolution}
Let $\Pi_s(\graphsix_R)$ the set of simple paths 
in $\graphsix_R$ from $s'$ to any vertex of $R$.
Similarly, let $\Pi_t(\graphsix_R)$ be the set of simple paths
from any vertex of $R$ to $t$.
We denote by $s(\pi)$ (resp. by $t(\pi)$) the source 
(resp. the sink) of a path $\pi$.
Let the longest path length distribution function in the $s$-side 
and the $t$-side be
$\check \Phi(x,\Vec{y})$ and $\tilde \Psi_(\Vec{y})$, respectively, for 
$\Vec{y}\in \mathbb{R}^R$ and $x$,
that are given by 
\begin{align*}
\check \Phi(x,\Vec{y})&= \Pr\left[\bigwedge_{\pi\in \Pi_s(\graphsix_R)}\sum_{e\in \pi}X_e\le x-y_{t(\pi)}\right],\\
\tilde \Psi_V(\Vec{y})&=\Pr\left[\bigwedge_{\pi \in \Pi_t(\graphsix_R)}\sum_{e\in \pi}X_e\le y_{s(\pi)}\right],
\end{align*}
where $\tilde \Psi_V(\Vec{y})$ is in Proposition~\ref{proposition:PsiMax}
Let $F_{\rm MAX}(x)$ be the longest path length distribution function
of $\graphfive_R$, 
We have
\begin{align*}
F_{\rm MAX}(x-|U|)=\int_{\Vec{y}\in\mathbb{R}^R} \check \Phi(x,\Vec{y})\tilde \psi_V(\Vec{y})d\Vec{y},
\end{align*}
where $\tilde \psi_V(\Vec{y})=\left(\prod_{uv\in R}\frac{\partial}{\partial y_{uv}}\right)\tilde \Psi_V(\Vec{y})$.
\end{proposition}
\begin{proof}
Remember that $\tilde \Psi_V(\Vec{y})$ represents the probability that 
the lengths of all of the $v-t$ paths for $v\in V$ are all
at most given by the corresponding value of $y_{uv}$ for $u\in U_v$,
where $U_v$ is the set of parents of $v$ in the transportation graph.
Then, $D_R \tilde \Psi_V(\Vec{y})=\tilde \psi_V(\Vec{y})$ is the density
at these path lengths are specified by $\Vec{y}$.
By the definition, we have
\begin{align*}
&F_{\rm MAX}(x-|U|)=\Pr\left[\bigwedge_{\pi\in \Pi_{\graphsix_R}}\left(\sum_{e\in \pi} X_e\le x\right)\right],
\end{align*}
where the argument of $F_{\rm MAX}(x-|U|)$ is subtracted by $|U|$
due to the $|U|$ edges with static length $1$ in $\graphsix_R$.
This is given by aggregating all 
probability density $\tilde \psi_V(\Vec{y})$ 
of the path length from $r\in R$ to $t$ with respect to $\Vec{y}\in \mathbb{R}^R$, which leads to
\begin{align*}
&\int_{\mathbb{R}^R}\Pr\left[\bigwedge_{\pi\in\Pi_s(\graphsix_R)}\left(\sum_{e\in \pi}X_e\le x-y_{t(\pi)}\right) \middle|\bigwedge_{\pi\in\Pi_t(\graphsix_R)} \left(\sum_{e\in \pi}X_e= y_{s(\pi)}\right)\right]\tilde \psi_V(\Vec{y})d\Vec{y}
\end{align*}
Now, since all edge lengths are mutually independent, this is equal to
\begin{align*}
&\int_{\mathbb{R}^R}\Pr\left[\bigwedge_{\pi\in\Pi_s(\graphsix_R)}\left(\sum_{e\in \pi}X_e\le x-y_{t(\pi)}\right)\right]\tilde \psi_V(\Vec{y})d\Vec{y} 
=\int_{\mathbb{R}^R}\check\Phi(x,\Vec{y})\tilde \psi_V(\Vec{y})d\Vec{y},
\end{align*}
by definition.
\end{proof}

By evaluating $\check \Phi(x,\Vec{y})$, the following 
lemma enables us to prove Theorem~\ref{th:undirectedhard}.
\begin{restatable}{lemma}{lemmaTransportationGraphWithTails}
\label{lemma:transportation_graph_with_tails}
For
$4|U|\le x\le 4|U|+1$, we have 
\begin{align*}
&\check\Phi(x,\Vec{y})=\frac{(4|U|)!4^{-|U|}}{|U|!(2|U|)!(2|U|-1)!} \int_{\mathbb{R}}\Unif_{4|U|}(x-4|U|-\xi)\frac{\partial}{\partial x} \Psi_U(\xi,\Vec{y})d\xi.
\end{align*}
\end{restatable}
\begin{proof} 
Consider the conribution 
of the cycle $\hat C$ in $\graphsix_R$, 
which we count the length of edges 
$\{s,\hat u\}~~(\hat u\in U^+\cup U^-)$
before $s$; the length of $|U|$ edges with 
static length $1$; and $2|U|$ i.i.d. uniformly random 
length edges in the cycle. 
The distribution function $\Phi^+(x,\Vec{y})$ 
for the longest path lengths
starting at $s$ is equal to
\begin{align*}
&\Phi^+(x,\Vec{y})=\int_{\mathbb{R}^{U^+\cup U^-}}\Phi_{\hat C}(x-3|U|+1,\Vec{z})\prod_{\hat u\in U^+\cup U^-}\unif^+(z_{\hat u}-w_u) d\Vec{z}, 
\end{align*}
where, for $u\in U$, we set $w_u=\max_{uv\in R}\{y_{uv}\}$ 
as the longest path length from $u\in U$ to $t$. 
Here $x$ is translated by $-3|U|+1$ since the 
longest path, from $\hat C$, takes $|U|$ edges 
with static length $1$ and $|2U|-1$ edges 
with i.i.d. uniformly over $[1,2]$ random lengths
excluding exactly one random length edge of $\hat C$.
By Proposition~\ref{proposition:ring}, the above is equal to
\begin{align*}
&\int_{\mathbb{R}^{U^+\cup U^-}}\mathrm{Vol}(K_{\hat C})\prod_{\hat u\in U^+\cup U^-}(x-3|U|+1-z_{\hat u})H(x-3|U|+1-z_{\hat u})\unif(z_{\hat u}-1-w_u) d\Vec{z}. 
\end{align*}
By setting $z'_{\hat u}=z_{\hat u}-1-w_u$,
the above is equal to
\begin{align*}
&\int_{\mathbb{R}^{U^+\cup U^-}}\mathrm{Vol}(K_{\hat C})\prod_{\hat u\in U^+\cup U^-}(x-3|U|-w_u-z'_{\hat u})H(x-3|U|-w_u-z'_{\hat u}) \unif(z'_{\hat u})d\Vec{z}'.
\end{align*}
We restrict the interval of integration into $[0,1]^{U^+\cup U^-}$ 
since the integrand is zero due to the factor $\unif(z'_{\hat u})$ for
$\hat u\in U^+\cup U^-$.
By renaming $z'_{\hat u}$ as $z_{\hat u}$ for simplicity, 
the above is equal to
\begin{align*}
&\int_{[0,1]^{U^+\cup U^-}}\mathrm{Vol}(K_{\hat C})\prod_{\hat u\in U^+\cup U^-}(x-3|U|-w_u-z_{\hat u})H(x-3|U|-w_u-z_{\hat u})d\Vec{z}.
\end{align*}
For $3|U|\le x \le 3|U|+1$, the above is equal to
\begin{align}
&\mathrm{Vol}(K_{\hat C})\prod_{\hat u\in U^+\cup U^-}\frac{1}{2}(x-|3U|-w_u)^2 H(x-3|U|-w_u),\nonumber \\
&=2^{-2|U|}\mathrm{Vol}(K_{\hat C})\prod_{u\in U}(x-3|U|-w_u)^4 H(x-3|U|-w_u), \label{form:Phiplus}
\end{align}
since there are exactly two $\hat u$'s in $U^+\cup U^-$. 
This is because there
is one to one correspondence between $u\in U$ and $u^+\in U^+$, and
also one to one correspondence between $u\in U$ and $u^-\in U^-$.
We compute a convolution between $\Phi^+(x,\Vec{y})$ and 
$\unif_{|U|}^+(x)$ for $s'-s$ path length, which gives
\begin{align*}
&\check\Phi(x,\Vec{y})=\int_{\mathbb{R}}\Unif_{|U|}^+(x-\xi)\frac{\partial}{\partial \xi}\Phi^+(\xi,\Vec{y})d\xi
\end{align*}
Since $\Phi^+(x,\Vec{y})$ is equal to (\ref{form:Phiplus}),
the above is transformed into
\begin{align*}
&\int_{\mathbb{R}}\Unif_{|U|}^+(x-\xi)2^{-2|U|}\mathrm{Vol}(K_{\hat C})\frac{\partial}{\partial \xi}\prod_{u\in U}(\xi-3|U|-w_u)^4 H(\xi-3|U|-w_u)d\xi.
\end{align*}
By the definition of $\Unif^+_{|U|}$, this is equal to
\begin{align*}
&\int_{\mathbb{R}}\Unif_{|U|}(x-|U|-\xi)2^{-2|U|}\mathrm{Vol}(K_{\hat C})\frac{\partial}{\partial \xi}\prod_{u\in U}(\xi-3|U|-w_u)^4 H(\xi-3|U|-w_u)d\xi.
\end{align*}
By setting $\xi'=\xi-3|U|$, this is transformed into
\begin{align*}
&\int_{\mathbb{R}}\Unif_{|U|}(x-4|U|-\xi')2^{-2|U|}\mathrm{Vol}(K_{\hat C})\frac{\partial}{\partial \xi}\prod_{u\in U}(\xi'-w_u)^4 H(\xi'-w_u)d\xi'.
\end{align*}
We rename $\xi'$ as $\xi$ in the following.
Notice that we must have $\xi-w_u\le 1$ and $4|U|\le x$ to have $\Unif_{|U|}(x-4|U|-\xi)$ nonzero.
Since $\Unif_{m}(x)=\frac{1}{m!}x^m$, the above is equal to, for $4|U|\le x\le 4|U|+1$,
\begin{align*}
&(4|U|)!2^{-2|U|}\mathrm{Vol}(K_{\check C}) \int_{\mathbb{R}}\Unif_{5|U|}(x-4|U|-\xi)\frac{\partial}{\partial \xi}\prod_{u\in U}H(\xi-w_u)d\xi,
\end{align*}
which is guaranteed by the commutativity of convolution
by taking the $4|U|$-th antiderivative of $\Unif_{|U|}(x-4|U|-\xi)$ and $4|U|$-th derivative of 
$\prod_{u\in U}(\xi-w_u)^4 H(\xi-w_u)$.
By the sifting property of the Dirac delta, the above is equal to
\begin{align*}
&(4|U|)!2^{-2|U|}\mathrm{Vol}(K_{\check C}) \int_{\mathbb{R}}\int_{\mathbb{R}^{U}}\Unif_{5|U|}(x-4|U|-\xi)\\
&\hspace*{5cm}\frac{\partial}{\partial \xi}\prod_{u\in U}H(\xi-z_u)\>\>\delta\left(z_u-\max_{uv\in R}\{y_{uv}\}\right)d\Vec{z}d\xi.
\end{align*}
Since $\frac{d}{dx}H(x)=\delta(x)$ and, for $u\in U$,
$\prod_{uv\in R}H(z_u-y_{uv})=H(z_u-\max_{uv\in R}\{y_{uv}\})$
we can transform the above into
\begin{align*}
&(4|U|)!2^{-2|U|}\mathrm{Vol}(K_{\check C}) \int_{\mathbb{R}}\int_{\mathbb{R}^U}\Unif_{5|U|}(x-4|U|-\xi)\frac{\partial}{\partial \xi}\\
&\hspace*{5cm}\prod_{u\in U}H(\xi-z_u)\frac{\partial}{\partial z_u}\prod_{uv\in R}H(z_u-y_{uv})d\Vec{z}d\xi.
\end{align*}
Then, we apply the commutativity of convolution again
so that we consider the $|U|$-th derivative of $\Unif_{5|U|}(x-4|U|-\xi)$ 
and $|U|$-th antiderivative of 
$\frac{\partial}{\partial \xi}\prod_{u\in U}H(\xi-z_u)$, which results in
\begin{align*}
&\frac{(4|U|)!}{|U|!}2^{-2|U|}\mathrm{Vol}(K_{\check C}) \int_{\mathbb{R}}\Unif_{4|U|}(x-4|U|-\xi)\\
&\hspace*{5cm}\frac{\partial}{\partial \xi}\prod_{u\in U}(\xi-z_u)H(\xi-z_u)\frac{\partial}{\partial z_u}\prod_{uv\in R}H(z_u-y_{uv})d\xi.
\end{align*}
Since $(\xi-z_u)H(\xi-z_u)=\Unif(\xi-z_u)$ for $0\le z_u\le \xi \le 1$,
this is equal to
\begin{align*}
&\frac{(4|U|)!}{|U|!}2^{-2|U|}\mathrm{Vol}(K_{\check C}) \int_{\mathbb{R}}\Unif_{4|U|}(x-4|U|-\xi)\\
&\hspace*{5cm}\frac{\partial}{\partial \xi}\prod_{u\in U}\Unif(\xi-z_u)\frac{\partial}{\partial z_u}\prod_{uv\in R}H(z_u-y_{uv})d\xi.
\end{align*}
By definition of $\Psi_U(x,\Vec{y})$ we can transform this into
\begin{align*}
&\frac{(4|U|)!}{|U|!}2^{-2|U|}\mathrm{Vol}(K_{\check C}) \int_{\mathbb{R}}\Unif_{4|U|}(x-4|U|-\xi)\frac{\partial}{\partial \xi}\Psi_U(\xi,\Vec{y})d\xi,
\end{align*}
where $K(\check C)=\frac{1}{(2|U|)!(2|U|-1)!}$,
proving the claim.
\end{proof}

For LPPDF-IID of positive edge lengths and SPPDF-IID of
negative edge lenghts, we state the following lemma.
By Corollary~\ref{corollary:undirected_positive_shortest_negative_longest} and 
Lemma~\ref{lemma:positiveLPPDFIID_negativeSPPDFIID},
we have Theorem~\ref{th:undirectedhard}. 
\begin{restatable}{lemma}{positiveLPPDFIIDnegativeSPPDFIID}
\label{lemma:positiveLPPDFIID_negativeSPPDFIID}
In case we assume that the edge lengths are all i.i.d. uniformly  
over $[1,2]$, LPPDF-IID is $\sharpp$-hard.
Similarly, if we assume that the edge lengths are all i.i.d. uniformly 
over $[-2,-1]$, SPPDF-IID is $\sharpp$-hard.
\end{restatable}
\begin{proof}
We assume that we could compute the value of 
the longest path length distribution function $F_{\rm MAX}(x)$ of
$\graphfive_R$, which gives the value of $\tilde \Psi_{\rm MAX}(x)$
which is already proved to be $\sharpp$-hard in the previous subsection.
Then, by Proposition~\ref{proposition:Phi_Psi_convolution} and Lemma~\ref{lemma:transportation_graph_with_tails},
we have that 
\begin{align*}
&F_{\rm MAX}(x-|U|)=\int_{\mathbb{R}^R} \check\Phi(x,\Vec{y})\tilde\psi_V(\Vec{y}-2\cdot\Vec{1}) d\Vec{y} \\
&=\int_{\mathbb{R}^R}\frac{(4|U|)!4^{-|U|}}{|U|!(2|U|)!(2|U|-1)!} \int_{\mathbb{R}}\Unif_{4|U|}(x-4|U|-\xi)\frac{\partial}{\partial x} \Psi_U(\xi,\Vec{y})d\xi \tilde\psi_V(\Vec{y}-2\cdot\Vec{1})d\Vec{y}\\
&=\frac{(4|U|)!4^{-|U|}}{|U|!(2|U|)!(2|U|-1)!}\int_{\mathbb{R}}\Unif_{4|U|}(x-4|U|-\xi)\frac{\partial}{\partial x}\tilde\Psi_{\rm MAX}(\xi-2)d\xi,
\end{align*}
which we applied Proposition~\ref{proposition:PsiMax} for the last
equation.
By replacing $x'=x-|U|$ and $\xi'=\xi-2$, we have
\begin{align*}
&F_{\rm MAX}(x')=\frac{(4|U|)!4^{-|U|}}{|U|!(2|U|)!(2|U|-1)!}\int_{\mathbb{R}}\Unif_{4|U|}(x'-(3|U|+2)-\xi')\frac{\partial}{\partial x}\tilde\Psi_{\rm MAX}(\xi')d\xi',
\end{align*}
Since we have $\tilde \Psi_{\rm MAX}(\xi)=A \xi^{|U|+|V|+|R|}$ 
for $0\le \xi\le 1$ 
and a constant $A$ by Proposition~\ref{proposition:constant}, by 
renaming $x$ and $\xi$ as $x$ and $\xi$, respectively, for simplicity,
we have
\begin{align*}
&F_{\rm MAX}(x)=\frac{(4|U|)!4^{-|U|}}{|U|!(2|U|)!(2|U|-1)!}\int_{0}^1\Unif_{4|U|}(x-(3|U|+2)-\xi)\frac{\partial}{\partial x}\tilde\Psi_{\rm MAX}(\xi)d\xi\\
&=\frac{(4|U|)!4^{-|U|}}{|U|!(2|U|)!(2|U|-1)!}\int_0^1\Unif_{4|U|}(x-(3|U|+2)-\xi)\frac{1}{|U|+|V|+|R|}A\xi^{|U|+|V|+|R|-1}d\xi.
\end{align*}
We then set $s=x-(3|U|+2)$ and consider the case $0\le s \le 1$. 
We have
\begin{align*}
&F_{\rm MAX}(s+3|U|+2)=\frac{(4|U|)!4^{-|U|}}{|U|!(2|U|)!(2|U|-1)!}\int_0^s\frac{(s-\xi)^{4|U|}}{(4|U|)!}\frac{A\xi^{|U|+|V|+|R|-1}}{|U|+|V|+|R|}d\xi\\
&=\frac{4^{-|U|}A}{|U|!(2|U|)!(2|U|-1)!(|U|+|V|+|R|)}\int_0^s(s-\xi)^{4|U|}\xi^{|U|+|V|+|R|-1}d\xi \\
&=\frac{(4|U|)!(|U|+|V|+|R|)!4^{-|U|}A}{|U|!(2|U|)!(2|U|-1)!}s^{5|U|+|V|+|R|}.
\end{align*}
Thus, if we could compute $F_{\rm MAX}(x)$ for $3|U|+2\le x \le 3|U|+3$,
we would have the value of $A$ by 
\begin{align*}
A=\frac{|U|!(2|U|)!(2|U|-1)!}{(4|U|)!(|U|+|V|+|R|)!4^{-|U|}}(x-(3|U|+2))^{-(5|U|+|V|+|R|)},
\end{align*}
which gives the value of $\tilde \Psi_{\rm MAX}(x)$, proving
the lemma for LPPDF-IID.

The claim for SPPDF-IID is clear by Proposition~\ref{proposition:minmaxsymmetry}.
\end{proof}

\subsection{Other Probability Distributions}
We generalize the above results to a class of probability distributions
that approximates the uniform distribution in a small interval.
That is, the other edge length distributions make the problem 
no easier because uniform distribution is ubiquitous regardless of 
the graph is directed or undirected, as long as the candidates of the optimal
path share the same edge number $k$. 
Lemma~\ref{lemma:Breakpoint-hard} establishes that $\sharpp$-hardness 
holds when the edge length density function has a jump from $0$.
\begin{restatable}{lemma}{lemmaBreakpointHard}
\label{lemma:Breakpoint-hard}
Given a graph $G=(V,E)$ with $n$ vertices,
where the optimal path always has exactly $k$ 
edges regardless of the realization of the random edge lengths. 
That is, $\pi\in \Pi_G$ always satisfies $|\pi|=k$,
where $\Pi_G$ is the set of $s-t$ paths in $G$.
For $\epsilon>0$, $\ell=\Omega(2^{-\poly(n)})$ and $a\in\mathbb{R}$, 
let $F\in \distfamily_{\ell}(I;I_\leftarrow)\cup\distfamily_{\ell}(I;I_\rightarrow)$,
where $I=[a,a+\epsilon]$.
Let $\Vec{X}\in\mathbb{R}^E$ be a random vector with components $X_e\iidsim F(x)$ for $e\in E$. 
For $\epsilon=O(\ell/(n+1)!)$,  
it is $\sharpp$-hard to compute a $(1\pm \frac{1}{2n!})$-approximation of
\begin{enumerate}[{\normalfont (\arabic*)}]
\item $F_{\rm MAX}^{\Vec{X}}(k(a+\epsilon))=\bar F_{\rm MIN}^{-\Vec{X}}(-k(a+\epsilon))$ 
if $F\in \distfamily_{\ell}(I;I_\leftarrow)$; or
\item 
$F_{\rm MIN}^{\Vec{X}}(ka)=\bar F_{\rm MAX}^{-\Vec{X}}(-ka)$ 
if $F\in \distfamily_{\ell}(I;I_\rightarrow)$.
\end{enumerate}
\end{restatable}
\begin{proof}
In the following, the proof is for $F\in\distfamily_{\ell}(I;I_\leftarrow)$. 
The proof for $F\in\distfamily_{\ell}(I;I_\rightarrow)$ is symmetry.
We focus on $F_{\rm MAX}^{\Vec{X}}(x)$ in a small interval $[a,a+\epsilon]$
so that $f(x)=\frac{d}{dx} F(x)$ is almost a uniform density there.
Without loss of generality, we concentrate on the case $a=0$;
for the other $a$, we obtain the case by considering $F_{\rm MAX}(x-ka)$.
Note that $f(x)$ for $x>\epsilon$ does not matter in the computation of 
$F_{\rm MAX}(k\epsilon)$ since $\Vec{X}\ge 0$.
By the assumption that $f(x)$ is Lipschitz by Definition~\ref{definition:localuniformity}, 
we further assume that $|\frac{d}{dx}f(x)|\le 1$,
since otherwese we use $F(Mx)$ instead of $F(x)$, where $M\ge |\frac{d}{dx}f(x)|$ for $x\in I$.
We truncate and scale $F(x)$ so that $\tilde U(x)=\int_0^x \tilde u(x) dx$ where
$\tilde u(x)=H(k\epsilon-x) \frac{1}{F(k\epsilon)} f(x)$.

Let $\Vec{Y}\in[0,\epsilon]^E$ (resp. $\tilde{\Vec{X}}\in[0,k\epsilon]$) 
be a random vector with components $Y_e\iidsim U_{k\epsilon}$ 
(resp. $\tilde X_e\iidsim \tilde U$) for all $e\in E$, 
where $U_{k\epsilon}(x)$ is the uniform distribution function over $[0,k\epsilon]$.
By taking $\epsilon$ sufficiently small, 
we claim that $\tilde u(x)$ is approximately equal to 
$u_{k\epsilon}(x)=\frac{d}{dx}U_{k\epsilon}(x)$, 
a constant function, since we assume $\frac{d}{dx}f(x)=F(k\epsilon)\frac{d}{dx}\tilde u(x)\le1$.
If the claim is correct, we have 
that $F_{\rm MAX}^{\Vec{X}}(k\epsilon)$ 
is a $(1\pm\frac{1}{2n!})$-approximation of  
$F_{\rm MAX}^{\Vec{Y}}(k\epsilon)$, 
which gives the answer of $\#$LE2H by Lemma~\ref{lemma:Uniform-hard}.

To prove the claim,
we evaluate 
\begin{align}
D(k\epsilon)\defeq \left|F_{\rm MAX}^{\tilde{\Vec{X}}}(k\epsilon)/F_{\rm MAX}^{\Vec{Y}}(k\epsilon)\right|
\end{align}
for $\epsilon=\frac{1}{k}\frac{\ell}{4n}\frac{1}{n!}$,
where $\ell=\Omega(2^{-\mathrm{poly}(n)})$ and $f(x)>\ell$.
Since $f(x)>\ell$ and 
$\left|\frac{d}{dx}f(x)\right|=\left|F(\epsilon)\frac{d}{dx}\tilde u(x)\right|<1$ 
for $0<x<\epsilon$
by assumption $F\in \distfamily_{\ell}(I;I_\rightarrow)$ 
we have that 
\begin{align}
\frac{1}{\epsilon}-\frac{\epsilon}{\epsilon\ell}=\frac{1}{\epsilon}-\frac{1}{\ell}\le \tilde u(x) \le \frac{1}{\epsilon}+\frac{\epsilon}{\epsilon\ell}=\frac{1}{\epsilon}+\frac{1}{\ell}.
\end{align} 
Therefore, $1-\frac{k\epsilon}{\ell} \le \frac{\tilde u(x)}{u_{k\epsilon}(x)} \le 1+\frac{k\epsilon}{\ell}$ since $u_{k\epsilon}(x)=\frac{1}{k\epsilon}$ for $0<x<k\epsilon$,
implying 
\begin{align}
\left(1-\frac{k\epsilon}{\ell}\right)^{n}F_{\rm MAX}^{\Vec{Y}}(x)\le F_{\rm MAX}^{\tilde{\Vec{X}}}(x)\le \left(1+\frac{k\epsilon}{\ell}\right)^{n}F_{\rm MAX}^{\Vec{Y}}(x).
\end{align}
By the linear approximation, we have $(1-\frac{k\epsilon}{\ell})^{n} \ge 1-n\frac{\epsilon}{k\ell}=1-\frac{1}{4n}\frac{1}{n!}$.
We also have 
\begin{align*}
(1+\frac{k\epsilon}{\ell})^{n}\le 1+2n\frac{k\epsilon}{\ell}=1+\frac{1}{2n}\frac{1}{n!},  
\end{align*}
since $(1+x)^m=1+mx+\sum_{i=2}^m \binom{m}{i} x^i$ and 
$\sum_{i=2}^m \binom{m}{i}x^i \le mx$ for $0\le x\le 1/m^2$.

We then make sure that the above is a polynomial-time reduction.
We claim that obtaining $O(\log \frac{1}{\ell}+n\log n)$ bits
of $F_{\rm MAX}^{\Vec{X}}(\epsilon)$ from the oracle
in the fixed-point binary number system is sufficient to  
compute $F_{\rm MAX}^{\Vec{Y}}(\epsilon)$.
We have that 
$\epsilon^{-|E|}F_{\rm MAX}^{\Vec{X}}(\epsilon)=F_{\rm MAX}^{\tilde{\Vec{X}}}(\epsilon)$, 
where $\epsilon$ 
needs $\log \frac{1}{\ell}+ n\log n$ bits before the most significant bit plus 
$\lceil \log n!\rceil =O(n\log n)$
significant bits for computing $n!$.
Since the multiplication finishes in polynomial time with respect to $n$,
the entire reduction finish in polynomial time.
Thus, computing $F_{\rm MAX}^{\Vec{X}}(k\epsilon)$ is $\sharpp$-hard. 
So is $\bar F_{\rm MIN}^{-k\Vec{X}}(-\epsilon)=F_{\rm MAX}^{\Vec{X}}(k\epsilon)$ 
by Proposition~\ref{proposition:minmaxsymmetry}.
\end{proof}

The following lemma extends the hard-to-compute interval into almost all support of 
the edge length distribution function.
\begin{restatable}{lemma}{lemmaBreakpointHardInterval}
\label{lemma:Breakpoint-hard-interval}
Let $\Pi_G$ be the set of $s-t$ paths in graph $G=(V,E)$,
where the optimal path always has exactly $k$ edges regardless of the realization 
of the random edge lengths. 
For $\epsilon>0$, $\ell=\Omega(2^{-\poly(n)})$ and $a,x_0\in\mathbb{R}$, 
let $J=\{a\}\cup[x_0,x_0+\epsilon]$
and $F\in\distfamily_{\ell}(J;J_\leftarrow)\cup\distfamily_{\ell}(J;J_\rightarrow)$. 
Let $\Vec{X}\in\mathbb{R}^E$ be a random vector with components $X_e\iidsim F$ for $e\in E$. 
For $\epsilon=O(\ell/(n+1)!)$,
it is $\sharpp$-hard to 
compute a $(1\pm \frac{1}{2n!})$-approximation of 
\begin{enumerate}[{\normalfont (\arabic*)}]
\item $F_{\rm MIN}^{\Vec{X}}(x_0)=\bar F_{\rm MAX}^{-\Vec{X}}(-x_0)$
if $F\in\distfamily_{\ell}(J;J_\leftarrow)$; or
\item $F_{\rm MAX}^{\Vec{X}}(x_0)=\bar F_{\rm MIN}^{-\Vec{X}}(-x_0)$ 
if $F\in\distfamily_{\ell}(J;J_\rightarrow)$.
\end{enumerate}
\end{restatable}
We include the point $a\in\mathbb{R}$ to $J$ to represent that 
$\inf J$ and $\sup J$ for $\distfamily_{\ell}(J;J_\leftarrow)$ and 
$\distfamily_{\ell}(J;J_\rightarrow)$ may be 
far from $x_0$. The proof is in the following.
\begin{proof}
We consider the case (1) $F\in\distfamily_{\ell}(J;J_\leftarrow)$ and
show that computing $F_{\rm MIN}^{\Vec{X}}(x_0)$ is $\sharpp$-hard,
assuming $x_0>a$ since the case $x_0\le a$ is clear by Lemma~\ref{lemma:Breakpoint-hard}. 
The proof for the case (2) $F\in\distfamily_{\ell}(J;J_\rightarrow)$ is symmetry.
As in before, 
we focus on the case $a=0$ without loss of generality.
In case (1), if there exists an oracle $M$ that gives 
a $(1\pm\frac{1}{2n!})$-approximation of 
$F_{\rm MIN}^{\Vec{X}}(x_0)$ for $x_0$,
we prove that we can compute $F_{\rm MAX}^{\Vec{X}^-}(-\epsilon)$ for
some i.i.d. random vector $\Vec{X}^-$, which is $\sharpp$-hard to compute
by Lemma~\ref{lemma:Breakpoint-hard}.

Since $\Vec{X}\ge \Vec{0}$, 
we remark that $f(x)$ for $x>x_0+\epsilon$ does not matter 
in the computation of $F_{\rm MIN}^{\Vec{X}}(x_0)$.
We truncate $F(x)$ so that $\tilde F(x)=\int_0^x \tilde f(x)dx$ where
$\tilde f(x)=H((x_0+\epsilon)/k-x)\frac{1}{F((x_0+\epsilon)/k)} f(x)$.
Let $\tilde{\Vec{X}}\in [0,(x_0+\epsilon)/k]^E$ be a random vector with components 
$\tilde X_e\iidsim \tilde F$ for $e\in E$.
It is clear that we can compute $F_{\rm MIN}^{\Vec{X}}(x_0)$ if we have
the value of $F_{\rm MIN}^{\tilde{\Vec{X}}}(x_0)$ since, by definition, 
\begin{align*}
F_{\rm MIN}^{\tilde{\Vec{X}}}(x_0)=(F((x_0+\epsilon)/k))^{-|E|}F_{\rm MIN}^{\Vec{X}}(x_0).
\end{align*}
Now, consider the translation $F^-(x)=\tilde F(x+(x_0+\epsilon)/k)$
corresponding to random vector $\Vec{X}^-=\tilde{\Vec{X}}-\frac{x_0-\epsilon}{k}\Vec{1}$. 
Then, we have $F^-\in \distfamily_{\ell}(I^-;I^-_\rightarrow)~~(I^-=[-(x_0+\epsilon)/k,0])$,
so that $F_{\rm MIN}^{\tilde{\Vec{X}}}(x_0)=\bar F_{\rm MAX}^{\Vec{X}^-}(-\epsilon)$, 
meaning that $F_{\rm MIN}^{\tilde{\Vec{X}}}(x_0)$ is $\sharpp$-hard 
to compute by Lemma~\ref{lemma:Breakpoint-hard}.
\end{proof}

We further extend the above so that we do not need
both of the upper and the lower bounds of the edge lengths.
\begin{restatable}{lemma}{lemmaInftyHard}
\label{lemma:infty_hard}
Let $\Pi_G$ be the set of $s-t$ paths in graph $G=(V,E)$,
where the optimal path always has exactly $k$ edges regardless of the realization 
of the random edge lengths. 
For $\ell=\Omega(2^{-\mathrm{poly}(n)})$, $\epsilon=O(\ell/(n+1)!)$ and $x_0\in \mathbb{R}$, 
let $F\in \distfamily_{\ell}(I;\emptyset)$, where $I=[x_0,x_0+\epsilon]$.
Let $\Vec{X}\in\mathbb{R}^E$ be a random vector with components $X_e\iidsim F$ for $e\in E$.
For a constant $c$,
it is $\sharpp$-hard to compute $(1\pm \frac{1}{cn!})$-approximation of
both $F_{\rm MAX}^{\Vec{X}}(x_0)$ and $ F_{\rm MIN}^{\Vec{X}}(x_0)$.
\end{restatable}
\begin{proof}
We prove the $\sharpp$-hardness of computing $F_{\rm MAX}^{\Vec{X}}(x_0)$.
The proof for $F_{\rm MIN}^{\Vec{X}}(x_0)$ is symmetry.
Let $F\in \distfamily_{\ell}(I;\emptyset)$.
By definition, 
there exists some $x_{\ell}$ such that $F(x)\ge 1-\ell$ for $x>x_{\ell}$.
Then, we truncate $F(x)$ such that 
$\tilde F(x)=\int_{-\infty}^{x}\tilde f(x)dx$ where
$\tilde f(x)=H(x_{\ell}-x)\frac{1}{F(x_{\ell})}f(x)$.

To prove that $\ell=\Omega(2^{-\poly(n)})$ is sufficiently small, consider the following.
Let $\tilde{\Vec{X}}\in(-\infty,x_\ell]^E$ 
be a random vector with components $\tilde X_e\iidsim \tilde F$.
Now, by Lemma~\ref{lemma:Breakpoint-hard-interval}, 
computing $F_{\rm MAX}^{\tilde{\Vec{X}}}(x)$ 
for $x<x_{\ell}$ is $\sharpp$-hard 
since $\tilde F\in \distfamily_{\ell}(J;J_\leftarrow)$, where
$J=\{x_\ell\}\cup[x_0,x_0+\epsilon]$.
We use 
$F_{\rm MAX}^{\Vec{X}}(x)$ as our approximation of 
$F_{\rm MAX}^{\tilde{\Vec{X}}}(x)$, which is
a polynomial-time reduction. 

Assume an oracle $M$ that outputs a value $A$, satisfying 
$1-\frac{1}{cn!}\le \frac{A}{F_{\rm MAX}^{\Vec{X}}(x_0)}\le 1+\frac{1}{cn!}$ for a constant $c$.
Then, we claim that 
\begin{align}
1-\frac{1}{cn!}\le \frac{F_{\rm MAX}^{\Vec{X}}(x_0)}{F_{\rm MAX}^{\tilde{\Vec{X}}}(x_0)}\le 1. \label{form:approximating_tilde}
\end{align}
If so, we have that 
$1-\frac{2}{cn!}\le (1-\frac{1}{cn!})^2\le \frac{A}{F_{\rm MAX}^{\tilde{\Vec{X}}}(x_0)}\le 1+\frac{1}{cn!}$,
implying that, for $c\ge 4$, $M$ computes $F_{\rm MAX}^{\tilde{\Vec{X}}}(x)$ sufficiently
accurately so that we have the exact answer of a $\sharpp$-hard problem.

It remains to prove (\ref{form:approximating_tilde}).  
Clearly, we have 
$F_{\rm MAX}^{\Vec{X}}(x_0)\le F_{\rm MAX}^{\tilde{\Vec{X}}}(x_0)$
for $x_0\le x_{\ell}$, since $f(x)\le \tilde f(x)$ by definition for $x\le x_{\ell}$.
Then, we focus on the earlier inequality of (\ref{form:approximating_tilde}).
We here set $\ell=\frac{1}{c|E|(n!)}$.
By definition, we have that 
\begin{align}
&F_{\rm MAX}^{\tilde{\Vec{X}}}(x_0)=\int_{\mathbb{R}^E}\Pr\left[\bigwedge_{\pi\in\Pi_G}\sum_{e\in \pi} x_e\le x_0 \right]\prod_{e\in E}\tilde f(x_e)d\Vec{x}, \\
&\le F(x_\ell)^{-|E|}\int_{\mathbb{R}^E}\Pr\left[\bigwedge_{\pi\in\Pi_G}\sum_{e\in \pi} x_e\le x_0 \right]\prod_{e\in E}f(x_e)d\Vec{x}= F(x_\ell)^{-|E|} F_{\rm MAX}^{\Vec{X}}(x_0).
\end{align}
The inequality is due to $f(x)\ge \tilde f(x)=0$ for $x> x_\ell$.
The above arguments amount to
\begin{align}
\frac{F_{\rm MAX}^{\Vec{X}}(x_0) }{F_{\rm MAX}^{\tilde{\Vec{X}}}(x_0)}\ge F(x_\ell)^{|E|} \ge (1-\ell)^{|E|}\ge 1-|E|\ell =1-\frac{1}{cn!},
\end{align}
by linear approximation, proving (\ref{form:approximating_tilde}) and hence the claim of the lemma. 
\end{proof}

Now, by Lemma~\ref{lemma:infty_hard}, we establish Theorems~\ref{th:directedhard_manydistributions} as follows.
\begin{proof}(of Theorem~\ref{th:directedhard_manydistributions})
The claim is clear by Lemma~\ref{lemma:infty_hard},
since we can restrict the number of edges as exactly three in the optimal path candidates
regardless of the edge length distribution in $\distfamily_{\ell}(I;\emptyset)$
for $\ell=\Omega(2^{-\poly(n)})$, $\epsilon=O(\ell/(n+1)!)$ and $x_0\in \mathbb{R}$.
\end{proof}

As for an undirected graph, we state the following.
\begin{lemma}
\label{lemma:undirectedhard_negativepositive}
For an undirected graph, SPPDF-IID (resp. LPPDF-IID) is $\sharpp$-hard if
the distribution of the edge lengths is in 
$\distfamily_{\ell}(I^+;I^+_\leftarrow)$ (resp. $\distfamily_{\ell}(I^-,;I^-_\rightarrow)$), 
where $I^+=[x_0,x_0+\epsilon]$ and $I^-=[-x_0-\epsilon,-x_0]$
for $x_0\ge\epsilon$, $\ell=\Omega(2^{-\poly(n)})$ and $\epsilon=O(\ell/(n+1)!)$.
\end{lemma}
\begin{proof}
Notice that we can scale $x$ by $1/\epsilon$ so that
we use $F(x/\epsilon)$ instead of $F(x)$,
which implies considering $F\in\distfamily_{\ell}([\epsilon,2\epsilon];(-\infty,\epsilon))$ is 
equivalent to considering $F\in\distfamily_{\epsilon^2\ell}([1,2];(-\infty,1))$.
Since our reductions used for Theorem~\ref{th:undirectedhard} 
can restrict the number of edges
in case edge lengths are all negative or all positive 
with their absolute values at least $\epsilon$
for undirected graphs, 
we apply Lemma~\ref{lemma:infty_hard} for 
the edge length distributions in $\distfamily_{\ell}(I^+;I^+_\leftarrow)$ 
or $\distfamily_{\ell}(I^-;I^-_\rightarrow)$
for $x_0\ge\epsilon$, $\ell=\Omega(2^{-\poly(n)})$ and $\epsilon=O(\ell/(n+1)!)$, 
which proves the lemma.
\end{proof}

For the proof of Theorem~\ref{th:undirectedhard_manydistributions},
we consider polynomially many parallel edges with i.i.d. random lengths, 
utilizing the following
\begin{align*}
\Pr\left[\max_{i=1,\dots,N}\{X_i\}\le x\right]&=
\Pr\left[\bigwedge_{i=1,\dots,N}X_i\le x\right]=F^N(x), ~~\text{and}\\
\Pr\left[\min_{i=1,\dots,N}\{X_i\}> x\right]&=
\Pr\left[\bigwedge_{i=1,\dots,N}X_i> x\right]=\bar F^N(x)= (1-F(x))^N,
\end{align*}
for $X_i~\iidsim F ~~(i=1,\dots,N)$.
Then, if the edge lengths are drawn from $F$, 
the problems SPPDF-IID and LPPDF-IID are reduced to 
the problems dealt in Lemma~\ref{lemma:undirectedhard_negativepositive},
which presented a special case in the following.
\begin{proof} (of Theorem~\ref{th:undirectedhard_manydistributions})
Consider LPPDF-IID with the edge length distribution function $F$ 
such that 
$F^N\in \distfamily_{\ell}(I^+;\emptyset)\cup \distfamily_{\ell}(I^-,(-\epsilon,\infty))$,
for $N=O(\poly(n))$.
We further focus on $F$ such that 
$F^N\in\distfamily_{\ell}(I^+;\emptyset)$
and make sure the existence of such $N$.
We assume that $F(x_\ell)\le 1-\frac{1}{\phi(n)}$ for $\phi(n)=\poly(n)$ and $\epsilon\le x_\ell \le x_0$.


Then, we replace each edge by $N$ parallel edges,
which is equivalent to consider the edge length distribution $F^N(x)$.
We consider truncating $F(x)$ at $x=x_\ell$
such that $\tilde F(x)=\int_{x_\ell}^x\tilde f(x) dx$ for
$x>x_\ell$ and $\tilde f(x)=H(x-x_\ell)\frac{1}{1-F(x_\ell)}f(x)$, where 
$f(x)=\frac{d}{dx}F(x)$.
Given any height-2 poset $\mathcal{P}=(U\cup U', R)$, 
consider related $\graphfive_R$ defined in Definition~\ref{definition:graphfive}.
Let $\graphfive_R^N$ be the graph given by 
replacing each edge of $\graphfive_R$ by 
$N=N_1N_2$ parallel edges, where $N_1=\poly(n)$ and 
$N_2=\lceil \phi(n)\ln2 \rceil$,
We claim that $F^{N_2}(x_\ell)\le \frac{1}{2}$,
which will be shown later.
If so, in $\graphfive_R^N$, LPPDF-IID in case edge lengths obey
the distribution given by $F(x)$ is equivalent to
the LPPDF-IID in case edge lengths obey $F^N(x)$ in $\graphfive_R$.
Then, notice that 
the longest path length distribution function 
$F_{uv}(x)$ in $\graphfive_R^N$ 
between any two adjacent vertices is bounded from above by 
\begin{align*}
F^{N_1N_2}(x_\ell)=\Pr\left[\bigwedge_{i=1,\dots,N} (X_i \le \epsilon)\right]\le 2^{-N_1},
\end{align*}
where $X_i\iidsim F$.
Thus, by truncating $F^N(x)$ at $x=\mu$ so that we obtain
$\tilde F^N(x)$ instead of $F^N(x)$, we can repeat the arguments
the LPPDF-IID in $\graphfive_R^N$ with the edge length distribution
function $F(x)$ sufficiently accurately approximates  
the LPPDF-IID with $\tilde F(x)$ at $x=x_0$. 
Since $\tilde F\in \distfamily_{\ell}(I^+;I^+_\leftarrow)~~(I^+=[x_0,x_0+\epsilon])$,
for some $x_0\ge \epsilon$,
the $\sharpp$-hardness of LPPDF-IID in $\graphfive_R^N$ is proved
by Lemma~\ref{lemma:undirectedhard_negativepositive}
which establishes the $\sharpp$-hardness of the case with edge length distribution
function $F(x)$.

It remains to prove that 
$N_2=\lceil \phi(n)\ln2 \rceil$ is sufficient
for $F^{N_2}(x)\le \frac{1}{2}$.
Remember that $F(x_\ell)< 1-2^{-\phi(n)}$.
By observing that $1+x\le e^x$ for $x=-\frac{1}{\phi(n)}$, we have
\begin{align*}
F^{N_2}(x_\ell)\le \left(1-\frac{1}{\phi(n)}\right)^{N_2}\le e^{-N_2/\phi(n)}\le e^{-\ln2} \le \frac{1}{2},
\end{align*}
proving the claim for LPPDF-IID.
The claim for $\sharpp$-hardness of SPPDF-IID is clear 
by Proposition~\ref{proposition:minmaxsymmetry}.
\end{proof}

\section{$\xp$ Algorithm for Directed Graphs with Edge Lengths i.i.d. Uniformly over $[0,1]$}
In this section, we focus on the case where edge lengths are 
i.i.d. uniformly distributed random variables.
Specifically, we show that SPPDF-IID can be 
computed in polynomial time if the underlying treewidth
is bounded by a constant, and so is LPPDF-IID by symmetry.
To apply the tree decomposition to our directed graph,
we consider the tree decomposition of the underlying undirected graph of $G$.
Our algorithm is based on a bottom-up dynamic programming on the tree decomposition,
maintaining the piecewise polynomial function as the state.
The basic idea is that the shortest and longest
length distribution functions can be expressed as a sequence of convolutions.
In the following, we focus on $F_{\rm MIN}(x)$ for the conciseness,
assuming that the edge lengths are given by 
$\Vec{X}\in[0,1]^E$ where $X_e\iidsim \Unif$ for $e\in E$.

\subsection{Polynomial Time Algorithm for Underlying Treewidth $k$ Graphs}
Given a directed graph $G=(V, E)$ and its tree decomposition with width $k$,
we consider how to compute $F_{\rm MAX}(x)$ and $F_{\rm MIN}(x)$ 
in the following.
Let ${\cal B}=\{B_1,\dots,B_b\}$ be
the set of bags of the tree decomposition.
We assume that $B_1$ is the root of the tree decomposition, and bags
are ordered in breadth-first search order from the root $B_1$.
We show the outline of our algorithm as follows 
\algorithmmainalg*
We will define the undefined words 'sentinel' and symbols $F_{\partial G_i}$ and $\partial_E[n]$ in the above later.
Since our algorithm heavily depends on how the recurrence
of $F_{\partial G_i}(\Vec{y}_{\partial_E[i]})$ is defined,
let us see the detail of the recurrence in the next.

\subsection{Constructing the Recurrence}
\label{section:recurrence}
We construct a recurrence to compute the shortest path length 
distribution function 
$F_{\rm MIN}(x)$ of the directed graph $G=(V,E)$.
Our algorithm is applicable to a graph with multiple sources and sinks.
Let $S(G)$ and $T(G)$ $(S(G)\cap T(G)=\emptyset)$ 
be the set of the sources and the sinks of $G$, respectively.
Let $\Pi_G$ be the set of all paths between $S(G)$ and $T(G)$.
We consider the following special vertex and its edges.
\definitionsentinel*
In the following, we assume that the sentinel vertex and edges are already 
added to $G$.
The edge $\bot v$ (resp. $v\bot$) represents the case 
$v$ has no incoming (resp. outgoing) edge in a certain computation, 
which will be explained later.

We will extend the idea of Theorem~\ref{th:DAGLP}~\cite{AOSY2009} 
to handle cycles.
The first idea for our recurrence is to associate, with each edge $uv$, 
two dummy variables $z_{uv,u}$ and $z_{uv,v}$, 
which are the components of $\Vec{z}_{\rm src}$ and $\Vec{z}_{\rm dst}$, respectively. 
These variables represents the shortest path length from the edge source, 
or from the edge destination, to the overall sink of the graph.
 In the following, we consider the following probability: 
\begin{align}
F_{G}(\Vec{z}_{E_\bot})=\Pr\left[\bigvee_{\pi\in \Pi_G}\left(\sum_{uv\in \pi}X_{uv} \le z_{e_s(\pi),s_\pi}-z_{e_t(\pi),t_\pi}\right)\right], \label{form:y}
\end{align}
where $e_s(\pi)$ and $e_t(\pi)$ are the first and the last edges of $\pi$;
$s_\pi\in S(G)$ and $t_\pi\in T(G)$ are the source and the sink of
path $\pi\in\Pi_G$. Here, the components 
$z_{e_s(\pi),s_\pi}$ and $z_{e_t(\pi),t_\pi}$ for $\pi\in \Pi_G$ of 
$\Vec{z}_{E_\bot}\in\mathbb{R}^{S(G)\cup T(G)}$ 
are defined for the path length from $s_\pi$ to $t_\pi$. 
Notice that we have the distribution function $F_{\rm MIN}(x)$ 
of the simple shortest path length if 
we set $z_{\bot s,\bot}=x$ for $s\in S(G)$,
and $z_{t\bot,\bot}=0$ for $t\in T(G)$.
By holding the parameters, we can
execute the convolution and extend the graph structure.

The following notations are necessary to argue 
the set of vertices that have already been processed and then
adding another vertex into the computation.
\begin{definition}
Let $G=(V,E)$ be a directed graph with vertex set $V=[i]$.
Let $[i]=\{1,\dots,i\}$ for $i\ge 1$.
Let $G_i=([i],E_i)$ be a subgraph of $G$ induced by $[i]$. 
Let $N(i)$ be the set of adjacent 
vertices of $i$.
Then, we set $E(i)=E^+(i)\cup E^-(i)$, where
$E^+(i)=\{iv \in E | v\in N_i\}$ and 
$E^-(i)=\{vi \in E | v\in N_i\}$.
\end{definition}
We use $\partial$ for the concepts that incorporates the frontier. 
Specifically, we define the following.
\begin{definition}
We write $\partial_V[i]$ to mean the frontier vertex set, 
the set of adjacent vertices of $[i]$ in $G$, excluding the vertices in $[i]$.
That is, $\partial_V[i]=\bigcup_{j\in [i]}N(j) \setminus [i]$.
We also define the frontier edge sets, 
$\partial_E^-[i]=\{uv \in E | u\in \partial_V[i], v\in [i]\}$ and
$\partial_E^+[i]=\{uv \in E | u\in [i], v\in \partial_V[i]\}$.
We set $\partial_E[i]=\partial_E^-[i]\cup \partial_E^+[i]$.
\end{definition}

We define the following subgraphs of $G$ to describe
the intermediate states of the ongoing computation.
\definitiongreatergraph*

We construct the shortest path length distribution function as follows.
\intermediateF*
For achieving such $F_{\partial G_i}(\Vec{z}_{\partial_E[i]})$, 
we encode the graph structure into 
a probability of boolean combinations 
of inequalities, where we use convolutions for 
combining the probabilities. 

We consider two forms as the gadgets for the vertices and edges.
\definitiongadgets*
For the vertex gadget, 
we take the partial derivative of the probability
$\Pr\left[\bigvee_{ui,iv\in E(i)}\mathcal{U}_{ui,iv}(\Vec{z}^+_{\rm src},\Vec{z}^-_{\rm dst})\right]$ 
with respect to each component of $\Vec{z}^-_{\rm dst}$.
To understand the behavior of the vertex gadget,
we remark the following.
\begin{remark}
\label{remark:generalizedfunction}
The vertex gadget $\mathcal{V}_i$ is an $|E(i)|$-variables generalized function, 
which is nonzero (the value is undefined like $\delta(0)$)
when $\bigwedge_{ui\in E^-(i)}\bigvee_{iv\in E^+(i)} \mathcal{U}'_{ui,iv}(\Vec{z}^+_{\rm src},\Vec{z}^-_{\rm dst})$ is true, where 
\begin{align*}
&\mathcal{U}'_{ui,iv}(\Vec{z}^+_{\rm src},\Vec{z}^-_{\rm dst})=\\
&\left(\vphantom{\frac{1}{1}} z_{ui,i}=z_{iv,i}\right)\wedge 
\left(\bigwedge_{u'i\in E^-(i)\setminus \{ui\}}(z_{u'i}=z_{i\bot})\right) \wedge
\left(\bigwedge_{iv'\in E^+(i)\setminus\{iv\}} (0\le z_{\bot i,i}-z_{iv',i})\right).
\end{align*}
\end{remark}
\begin{proof}
We can understand $\mathcal{V}_i$ by approximating 
$\Pr\left[\bigvee_{ui,iv\in E(i)}\mathcal{U}_{ui,iv}(\Vec{z}^+_{\rm src},\Vec{z}^-_{\rm dst})\right]$ with an $|E(i)|$-variables function
by the analogy that the Dirac delta is a narrowed uniform density. 
Let $\Vec{X}\in [0,1]^{E(i)}$ be a random vector with its components $X_{ui,iv}\iidsim\Unif$
for $ui,iv\in E(i)$.
Then, we have $\mathcal{V}_i=\lim_{\epsilon\rightarrow+0} \mathcal{V}_i^{(\epsilon)}$ where
we definie $\mathcal{V}_i^{(\epsilon)}$ by the following.
\begin{align*}
&\mathcal{V}_i^{(\epsilon)}(\Vec{z}^+_{\rm src},\Vec{z}_{\rm dst}^-)=D_{\Vec{z}^-}\Pr\left[\bigvee_{ui,iv\in E(i)}\mathcal{U}^{(\epsilon)}_{ui,iv}(\Vec{z}^+_{\rm src},\Vec{z}^-_{\rm dst})\right],\\
&\text{where}~~\mathcal{U}^{(\epsilon)}_{ui,iv}(\Vec{z}^+_{\rm src},\Vec{z}^-_{\rm dst})=(\epsilon X_{ui,iv}\le z_{ui,i} - z_{iv,i} )\wedge\mathrm{Garbage}^{(\epsilon)}(ui,iv,\Vec{z}),\\
&\text{and}~~\mathrm{Garbage}^{(\epsilon)}(ui,iv,\Vec{z})=\bigwedge_{\begin{subarray}{c}
u'i,iv'\in E(i)\\
u'\neq u, v'\neq v
\end{subarray}}(\epsilon X_{u'i,i\bot}\le z_{u'i,u'}-z_{i\bot,i}\wedge \epsilon X_{\bot i, iv'}\le  z_{\bot i,i}-z_{iv',i}). 
\end{align*}
Clearly, $\mathcal{V}_i^{(\epsilon)}$ exists for $\epsilon>0$.
So is $\mathcal{V}_i$ in the sense of generalized function\footnote{
Pragmatically, any $\epsilon>0$ that is sufficiently close to $0$ will work approximately.
}. 
\end{proof}

In the following, we consider the construction in a bottom-up manner.
Our $F_{\partial G_i}$ has $|\partial_E[i]|$ arguments, 
which are necessary to describe the path length from the source or 
the indegree zero vertices to the sink or the outdegree zero vertices
in $\partial G_i$. 
We remind that the vertices on the frontier of $\partial G_i$
come from the edges of $G$.
The following definition eliminates $z_{uv,u}$ for any $uv\in E$, 
where a vertex gadget of $i$ is left with $z_{uv,v}$.
\begin{definition}
\label{definition:FGi}
  For $i=1,\dots,n$ and $\Vec{z}^+_{\rm dst}\in \mathbb{R}^{E^+(i)}, \Vec{z}^-_{\rm dst}\in \mathbb{R}^{E^-(i)}$, we set 
\begin{align*}
\mathcal{V}^+_i(\Vec{z}_{\rm dst}^+,\Vec{z}_{\rm dst}^-)= \hspace*{-6mm}
\int\limits_{\hspace*{7mm}\Vec{\zeta}_{\rm src}\in \mathbb{R}^{E^+(i)}}\hspace*{-8mm}
\mathcal{V}_i(\cat{\Vec{\zeta}_{\rm src}}{\Vec{z}_{\rm dst}^-})\left(
\mathcal{E}_i(\cat{\Vec{\zeta}_{\rm src}}{\Vec{z}_{\rm dst}^+})\right) d\Vec{\zeta}.  
\end{align*}
\end{definition}
By this convolution, we obtain 
$\mathcal{V}_i^+(\Vec{z}^+_{\rm dst},\Vec{z}^-_{\rm dst})$ by replacing,
in the definition of 
$\mathcal{V}_i(\Vec{z}^+_{\rm src},\Vec{z}^-_{\rm dst})$, 
each $z_{iv,i}$ by $z_{iv,v}-X_{iv}$ that comes from $\mathcal{E}_i$.
\begin{proposition}
\label{proposition}
\begin{align*}
  &\mathcal{V}^+_i(\Vec{z}^+_{\rm dst},\Vec{z}^-_{\rm dst})=D_{\Vec{z}^-}\Pr\left[\bigvee_{ui,iv\in E(i)}\mathcal{U}_{ui,iv}(\Vec{z}^+_{\rm dst},\Vec{z}^-_{\rm dst})\right], \\
  &\text{where}~~\mathcal{U}_{ui,iv}(\Vec{z}^+_{\rm dst},\Vec{z}^-_{\rm dst})=(X_{iv}\le z_{ui,i} - z_{iv,v} )\wedge\mathrm{Garbage}(ui,iv,\Vec{z}),\\
  &\text{and}~~\mathrm{Garbage}(ui,iv,\Vec{z})=\bigwedge_{\begin{subarray}{c}
  u'i,iv'\in E(i)\\
  u'\neq u, v'\neq v
  \end{subarray}}(0\le z_{u'i,u'}-z_{i\bot,i}\wedge X_{v'}\le  z_{\bot i,i}-z_{iv',v}) 
\end{align*}
\end{proposition}

Remember that $(u,i+1),(i+1,v)\in \partial_E[i]~(u,v\in [i])$ are 
vertices in $\partial G_i$; the following identifies them as one vertex $i+1$ 
and gives an intermediate $F_{\partial G_i}(\Vec{z}_{\partial_E[i+1]})$.
\begin{definition}
\label{definition:FGi+1}
  Let  $\Delta E^\pm(i+1)=E^\pm(i+1)\cap \partial_E^\pm[i]$, taking the signs respectively.
  Also, let $\Delta E(i+1)=\Delta E^+(i+1)\cup \Delta E^-(i+1)$. 
  We have 
  \begin{align*}
    F_{\partial G_i}(\Vec{z}_{\partial_E[i+1]})
    =\hspace*{-6mm}\int\limits_{\hspace*{3mm}\Vec{\zeta}\in \mathbb{R}^{\Delta E(i+1)}}\hspace*{-4mm}F_{\partial G_i}(\hat{\Vec{\zeta}})
    \mathcal{V}^+_{i+1}\left(\check{\Vec{\zeta}}\right)d\Vec{\zeta},
  \end{align*}
  where $\hat{\Vec{\zeta}}=\cat{\Vec{\zeta}}{\Vec{z}_{\partial_E[i]\setminus E_{i+1}}}$ and
  $\check{\Vec{\zeta}}=\cat{\Vec{\zeta}}{\Vec{z}_{E(i+1)\setminus \Delta E(i+1)}}$.
\end{definition}

To make sure the correctness of the recurrence, we prove the following lemma.
\lemmaFDN*
\begin{proof}
The proof is by induction on $i\in V$.
As the induction hypothesis, we assume
\begin{align} 
  F_{\partial G_i}(\Vec{z}_{\partial_E[i]})&=\Pr\left[\bigvee_{\pi\in \Pi_i}\bigvee_{p\in \mathrm{Con}(\pi)}\left(\sum_{e\in p} Z_e\le z_{e_s(p),s_p}-z_{e_t(p),t_p}\right)\right],
  \label{form:induction_hypothesis}
\end{align}
where $\mathrm{Con}(\pi)$ is the set of connected components of $\pi$ in $\partial G_i$.
For $p\in \mathrm{Con}(\pi)$, we define a boolean predicate
\begin{align*}
\mathcal{W}_p=\left(\sum_{e\in p} Z_e\le z_{e_s(p),s_p}-z_{e_t(p),t_p}\right),
\end{align*}
so that (\ref{form:induction_hypothesis}) is equivalent to 
$F_{\partial G_i}(\Vec{z}_{\partial_E[i]})=\Pr\left[\bigvee_{\pi\in \Pi_i}\bigvee_{p\in \mathrm{Con}(\pi)}\mathcal{W}_p\right]$.
All possible functional subgraphs in 
$\partial i=(\{i\}\cup E_i), \{(i,e^+),(e^-,i)| e^+\in E^+(i), e^-\in E^-(i)\})$,
are enumerated in $\mathcal{V}_i$'s condition, 
where edge $(i,e^+)$ has edge length $X_e$; 
edge $(e^-,i)$ has zero edge length.
Since $\partial i$ is a star graph for $i=1,\dots,n$,
the functional graph set $\pfs_1$ 
of $\partial 1$ coincides with $\Pi_1$. Therefore, 
(\ref{form:induction_hypothesis}) holds trivially, for $i=1$.
Also, remember Definition~\ref{definition:sentinel}.
the sentinel edges $\bot s$ and $t\bot$ for $s\in S(G)$ and $t\in T(G)$ have
length zero and the other sentinel edges have sufficiently large length $M$,
implying that the sentinel edges eliminate the paths that do not start 
from $S(G)$, or ends at $T(G)$. 

We proceed to the induction step.
Assume that we have $F_{\partial G_i}(w,\Vec{z}_{\partial_E[i]})$ properly 
computed so that the induction hypothesis (\ref{form:induction_hypothesis}) holds.
Then, $\mathcal{V}_{i+1}^+\left(\check{\Vec{\zeta}}\right)~~(\check{\Vec{\zeta}}=\cat{\Vec{\zeta}}{\Vec{z}_{E(i+1)\setminus \Delta E(i+1)}}, \Vec{\zeta} \in\mathbb{R}^{\Delta E(i+1)})$ 
represents the density at the point where 
the length sum from the source 
to each of $\Delta E(i+1)$ are arbitrarily close to the value specified by 
$\hat{\Vec{\zeta}}=\cat{\Vec{\zeta}}{\Vec{z}_{E(i+1)\setminus\Delta E(i+1)}}$.
By taking the product with 
$F_{\partial G_i}\left(\cat{\Vec{\zeta}}{\Vec{z}_{\partial_{E}[i+1]\setminus \Delta E(i+1)}}\right)$, 
each component $\zeta_{(u,i+1),i+1}$ of $\Vec{\zeta}\in \mathbb{R}^{\Delta E(i+1)}$ 
in the inequality condition of $\mathcal{V}_{i+1}^+$ is substituted by some path length:
\begin{enumerate}
\item
$\zeta_{(u,i+1),i+1}$ in 
$\sum_{e\in p} Z_e\le z_{e_s(p),s_p}-\zeta_{(u,i+1),i+1}$ (due to $\mathcal{W}_p$)
is replaced by $Z_{(i+1,v)}+z_{(i+1,v),v}$ 
in $\mathcal{V}^+_i$'s $Z_{(i+1,v)}\le \zeta_{(u,i+1),i+1}-z_{(i+1,v),v}$,
meaning that $(i+1,v)$ is a new sink, or 
\item
$\zeta_{(i+1,v),v}$ in 
$\sum_{e\in p} Z_e\le \zeta_{(i+1,v),v}- z_{e_t(p),t_p}$ (due to $\mathcal{W}_p$),
is substituted by $z_{(u,i+1),i+1}-Z_{(i+1,v)}$ 
that originates in $\mathcal{V}_{i+1}^+$'s condition 
$Z_{(i+1,v)}\le z_{(u,i+1),i+1}-\zeta_{(i+1,v),v}$
meaning that $(u,i+1)$ is a new source. 
\end{enumerate}
The same applies to each combination of edges 
$(u,i+1),(i+1,v)\in E(i+1)$ and $p\in \mathrm{Con}(\pi)~(\pi\in\Pi_i )$.
If the inequality of (\ref{form:induction_hypothesis}) is
$\sum_{e\in p} X_e\le\zeta_{(u,i+1),i+1}-\zeta_{(i+1,v),v}$,
then, a path $p_1$ having sink $(u,i+1)$ and the another path
$p_2$ having source $(i+1,v)$ are connected at vertex $i+1$,
resulting in replacing $\zeta_{(u,i+1),i+1}-\zeta_{(i+1,v),v}$ by
$X_{(i+1,v)}$ due to 
$\mathcal{V}_{i+1}^+$'s $X_{(i+1,v)}\le \zeta_{(u,i+1),i+1}-\zeta_{(i+1,v),v}$.
In case $p_1=p_2$, we have a cycle, which ends up with zero probability.

By expanding the condition in the probability into 
the DNF of $\mathcal{U}_{(u,i+1),(i+1,v)}$'s, 
each term corresponds to a functional subgraph $\pi\in \pfs_{i+1}$
of $\partial G_{i+1}$. 
This is because, 
in case of a term where it has a mismatch 
in the choice of edges at the two ends of an edge, 
the path length is increased by $M$ due to the sentinel edge.

Formally, given Definition~\ref{definition:gadgets},
by specifying a term of the DNF of $\mathcal{U}_{(u,i+1),(i+1,v)}$, 
we have a corresponding projections $P^+:[i+1]\mapsto E^+(i+1)$ and
$P^-:[i+1]\mapsto E^-(i+1)$.
Then, if there is a mismatch in the term 
at the ends of edge $uv$ if $P^+(v)\neq P^-(u)$,
``$\mathrm{Garbage}$'' in the vertex gadget  
substitutes $z_{ui,i}$ or $z_{iv, v}$
with $z_{i\bot,i}$ or $z_{\bot i, \bot}$
leading to the sufficiently large path length $M$ due to the sentinel edges.
Thus, the subgraph corresponding to 
the term has a path with its path length 
sufficiently large, implying that 
the term with a mismatch is always false.

Since any $\pi\in\pfs_{i+1}$ with nonzero probability 
has vertices with indegree and outdegree at most one and 
$\pi$ has no cycle, $\pi$ is a path in $\partial G_{i+1}$,
proving the induction hypothesis 
(\ref{form:induction_hypothesis}) for $\partial G_{i+1}$.
After obtaining $F_{\partial G_n}(\Vec{z}_{\partial_E[n]})$, 
we set $z_{\bot s,s}=x$ and $z_{s \bot,\bot}=0$ for $s\in S(G)$,
and $z_{v\bot,\bot}=z_{\bot v,v}=0$ for $v\in V\setminus S(G)$.
The value of $\Vec{z}_{\partial_E[i]}$ 
gives $x$ on the right hand side of $\mathcal{W}_p$.
Therefore, we have the lemma.
\end{proof}

By evaluating the running time of Algorithm~\ref{algorithm:mainalg}, We prove Theorem~\ref{th:XPalgorithm} as follows.

\begin{proof} (of Theorem~\ref{th:XPalgorithm})
Consider how many words we need to store the representation of
the ongoing computation.
Though what we actually compute is $f_{\partial G_i}(\Vec{z}_{\partial_E[i]})$,
we instead consider the amount of memories that is necessary for 
$F_{\partial G_i}(\Vec{z}_{\partial_E[i]})$.
We will explain that taking the derivative of 
$F_{\partial G_i}(\Vec{z}_{\partial_E[i]})$ may produce 
tolerable amount of piecewise polynomial terms.
By the recurrence given in Definitions~\ref{definition:FGi+1},
we have that the resulting
form of $F_{\partial G_i}(\Vec{z}_{\partial_E[i]})$ is a piecewise smooth
polynomial of components of $\Vec{z}_{\partial_E[i]}$.

Let $b$ be the number of bags in the tree decomposition.
We apply Observation~\ref{observation:integral} to estimate
the number of cases due to the step function with 
$\Vec{z}_{\partial_E[i]}$ in the argument.
Notice that there are at most $n^{(k+1)^2}$ cases with respect to
$\Vec{z}_{\partial_E[i]}$ for each bag $B$ satisfying 
$B\cap \partial_V[i]\neq \emptyset$.
To verify this,
consider the resulting form of 
$F_{\partial G_i}(\Vec{z}_{\partial_E[i]})$
into a sum of products so that
each term is in the form of $g(x)$ in
Observation~\ref{observation:integral}.
We claim the following about the step function factors' arguments:

\smallskip
\noindent{\bf Claim:}
If we expand $F_{\partial G_i}(\Vec{z}_{\partial_E[i]})$ into
the sum of products, the step function factor of each term
is the product of 
$H(\pm z_{uv,v}\pm j)$ and $H(\pm(z_{e_s(p),s_p}-z_{t(p),t_p})\pm j)$ for $uv\in \partial_E[i]$, 
a connected path $p\in\mathrm{Con}(\pi)$ for $\pi\subseteq T\in \Pi_{i}$ and $j\in\{0,1,\dots,n-1,M\}$,
where $H(x)$ is the step function as defined in the preliminaries section.
\smallskip

We prove the correctness of the Claim by induction.
As the base case, by Definition~\ref{definition:gadgets},
each term of $F_{\partial_E[1]}(\Vec{z}_{\partial_E[1]})=\mathcal{V}_1^+(\Vec{z}_{\partial_E[1]})$ is given by
a combination of $H(z_{u1,1}-z_{1v,v})$'s for $u1, 1v\in E(1)$ when expanded into 
the sum of products, implying the Claim for $i=1$.

In the induction step, assume that the claim holds for some $i$.
Consider, focusing on the arguments of the step function factors,
the computation of $F_{\partial G_i}(\Vec{z}_{\partial_E[i+1]})$.
We compute $F_{\partial G_{i+1}}(\Vec{z}_{\partial_E[i+1]})$
by taking the convolution with 
$F_{\partial G_i}(\Vec{z}_{\partial_E[i+1]})$
and $\mathcal{V}^+_{i+1}(\Vec{z}_{E(i+1)})$ for 
$uv\in E(i+1)\cap \partial_E[i]$.
Since there is only one outgoing edge of $i+1$,
by the definition of the uniform density, executing the integral 
may replace $z_{uv,v}$ by another $z_{vw,w}$ or $z_{vw,w}-1$,
increasing $j$ in the claim by one.
Since we process the vertices in reverse BFS order in
the tree decomposition in our algorithm,
$|\partial_E[j]|$ is at most the number of edges in a bag
per a subtree rooted at $B_j\in \mathcal{B}$,
meaning that we have at most $(k+1)k$ edges on the frontier 
for each connected component of $\partial G_i$ when
we have $F_{\partial G_i}(\Vec{z}_{\partial_E[i]})$.
By the above claim, the number of possible permutations of
upper limits and lower limits is at most $i^{(k+1)k}$
for each connected component.

Consider, now, computing the derivative of $F_{\partial G_i}(\Vec{z}_{\partial_E[i]})$.
Taking the derivative turns some polynomial into another polynomial with 
less degree, and also it turns some step function factors into some Dirac delta factors.
We can manage the step function factors and the Dirac delta factors by  
having one more bit memory space for each factor, showing whether the factor is
the step function or the Dirac delta.
Since there are at most $(k+1)k$ step function factors in 
$F_{\partial G_i}(\Vec{z}_{\partial_E[i]})$,
the differentiation operation makes at most $2^{(k+1)k}$ times more
terms in the representation on the memory.
  
Therefore, for each connected component, we may have to 
store at most $i^{(k+1)k}$ coefficients of a polynomial, 
and there are at most $i$ intervals 
($[j,j+1]$ for $j=0,\dots,i-1$)
due to the integration of the uniform distribution function.
Now, it takes at most $O(b i^{2(k+1)k}2^{(k+1)k})$ time to compute
$F_{\partial G_{i+1}}(\Vec{z}_{\partial_E[i+1]})$
from $F_{\partial G_i}(\Vec{z}_{\partial_E[i]})$
for $i=1,\dots,n$, which proves the theorem. 

Since our recurrence holds for the case where the edge length vector is 
$-\Vec{X}$, our algorithm can be applied for computing 
$F_{\rm MAX}^{\Vec{X}}(x)=\bar F_{\rm MIN}^{-\Vec{X}}(-x)$
by Proposition~\ref{proposition:minmaxsymmetry},
which proves the theorem for LPPDF-IID.
\end{proof}

In the proof of Theorem~\ref{th:XPalgorithm}, we used 
the following observation about the integration of step function products.
Remember that $H(x)=1$ for $x\ge 0$, otherwise $H(x)=0$.
Let $F(x)$ be a function $F: \mathbb{R}\mapsto \mathbb{R}$, whose derivative
is $f(x)$. Let $[i]=\{1,\dots,i\}$ for any $i\in \mathbb{N}~~([0]=\emptyset)$.
Then, for 
\begin{align}
  g(x)=f(x)\prod_{i\in[\ell]}H(x-\alpha_i) \prod_{i\in[m]}H(\beta_i-x),\nonumber
\end{align}
where $\alpha_1,\dots,\alpha_\ell$ and $\beta_i,\dots,\beta_m$ are real numbers,
we have
\begin{align}
  &\int\limits_{\mathbb{R}}\!\!g(x)dx\!=\![F(x)]^{\min_{i\in[m]}\{\!\beta_i\!\}}_{\max_{i\in[\ell]}\{\!\alpha_i\!\}}\!H\!\left(\!\min_{i\in[m]}\{\!\beta_i\!\}\!-\!\max_{i\in[\ell]}\{\!\alpha_i\!\}\!\!\right). \nonumber
\end{align}
In this integration we define the following.
\begin{definition}
  For $\int_{\mathbb{R}}g(x)dx$ in the above, $\alpha_1,\dots,\alpha_{\ell}$
  (resp. $\beta_1,\dots,\beta_{m}$) are
  the lower (resp. the upper) limits of the integral. 
\end{definition}
We can eliminate the $\max$ and $\min$ operations by
taking the following sum for all possible permutations of
the lower limits $\alpha_1,\dots,\alpha_\ell$ and the upper limits
$\beta_1,\dots,\beta_m$.
\begin{observation}
  \label{observation:integral}
Let $\Pi[i]$ be the set of all permutations of $[i]$.
Then, $\int_\mathbb{R}g(x) dx$ is, almost everywhere, equal to
\begin{align}
  &\hspace*{-16mm}\sum_{\hspace*{15mm}(p,q)\in \Pi[\ell]\times \Pi[m]}\hspace*{-15mm}(F(\beta_{q(1)})\!-\!F(\alpha_{p(\ell)}))H(\beta_{q(1)}-\alpha_{p(\ell)}) \hspace*{-5mm}\prod_{\hspace*{5mm}i\in[\ell-1]}\hspace*{-5mm}H(\alpha_{p(i+1)}-\alpha_{p(i)})\hspace*{-5mm}\prod_{\hspace*{5mm}i\in[m-1]}\hspace*{-5mm}H(\beta_{q(i+1)}-\beta_{q(i)}).\nonumber
\end{align}
\end{observation}

\section{Conclusion and Future Work}
In this paper, we proved that SPPDF-IID and LPPDF-IID are $\sharpp$-hard.
If the underlying undirected graph of the input graph has bounded treewidth, 
we showed an $\xp$ algorithm.
At this point, whether SPPDF-IID and LPPDF-IID are fixed parameter tractable
or $\#\mathrm{W}[1]$-hard is an open problem.

\section{Supplementary Materials}
\label{section:supplementary}
\propositionConvolution*
\begin{proof}
By replacing $x_1=nh$ for $h$ arbitrarily close to zero, we have
\begin{align}
  \Pr[X_1+X_2\le x] =\lim_{h\rightarrow 0}\sum_{n\in \mathbb{Z}}\Pr[hn+X_2\le x | hn\le X_1\le h(n+1)]~ \Pr[hn\le X_1\le h(n+1)].\nonumber
\end{align}
In order to transform the limit of the sum into a Riemann integral,
we multiply $h$ by the earlier probability and $1/h$ to the latter probability,
so that $\Pr[X_1+X_2\le x]$ is equal to
\begin{align*}
  &\lim_{h\rightarrow 0}\sum_{n\in \mathbb{Z}}h\Pr[hn+X_2\le x | hn\le X_1\le h(n+1)] \frac{1}{h}\Pr[hn\le X_1\le h(n+1)]\\
  &=\lim_{h\rightarrow 0}\sum_{n\in \mathbb{Z}}h\Pr[hn+X_2\le x | hn\le X_1\le h(n+1)] \frac{1}{h}(F_1(h(n+1))-F_1(hn)).
\end{align*}
By the definition of the derivative $f(x)=\frac{d}{dx}F(x)=\lim_{h\rightarrow 0}\frac{F(x+h)-F(x)}{h}$, the above is equal to
\begin{align*}
  &\lim_{h\rightarrow 0}\sum_{n\in \mathbb{Z}}h\Pr[hn+X_2\le x | hn\le X_1\le h(n+1)] f(x) \\
  &=\int_{\mathbb{R}}\Pr[x_1+X_2\le x | X_1=x_1]f_1(x_1) dx_1
\end{align*}
by the definition of the Riemann integral $\lim_{h\rightarrow 0}\sum_{n\in \mathbb{Z}} hf(nh)=\int_{\mathbb{R}}f(x) dx$.
Since $X_1$ and $X_2$ are mutually independent,
\begin{align*}
\int_{\mathbb{R}}\Pr[x_1+X_2\le x]f_1(x_1) dx_1
=\int_{\mathbb{R}}F_2(x-x_1)f(x_1)dx_1
\end{align*}
is equal to $\Pr[X_1+X_2\le x]$.
\end{proof}

\thDAGLP*
\begin{proof}
Let us consider an intermediate variables $z_v$ for each vertex $v\in V$.
Here, $z_v$ represents the longest path length 
between vertex $v$ and any sinks, outdegree $0$ vertices. 
Then, we can rewrite $F_{\rm MAX}(x)$ as 
\begin{align}
\Pr\left[\bigwedge_{\pi\in\Pi}\sum_{e\in\pi} X_e\le x\right]&=\Pr\left[\bigwedge_{u\in V\setminus T}\bigwedge_{v\in V_u}(X_{uv}+z_v\le z_u) \right] \nonumber \\
&=\Pr\left[\bigwedge_{u\in V\setminus T}\left(\max_{v\in V_u}\{X_{uv}+ z_v\}\le z_u\right) \right], \label{form:transform_max}
\end{align}
where we set $z_s=x$ (resp. $z_t=0$) for any in-degree zero vertex $s$ 
(resp. any out-degree zero vertex $t$).
The equality is justified by the algorithm for the longest path problem 
with static edge lengths, given a realization of edge length $\Vec{X}$.

We here consider the relation between $X_{uv}$, $z_u$, and $z_v$ for $uv\in E$.
The condition of the probability (\ref{form:transform_max}) holds if 
$z_u$ and $z_v$ satisfy $z_u=\max_{v\in V_u}\{X_{uv}+z_v\}$ for all $uv\in E$. 
We consider the probability that 
another random variable $W_u=\max_{v\in V_u}\{X_{uv}+z_v\}$ 
is very close to $z_u$, where 
$z_u\le W_u\le z_u+h_u$ is satisfied for $h_u>0$, 
arbitrarily close to $0$~($u\in V$).
Let $I\subset V$ be the set of all internal vertices,
that is, both in-degree and out-degree nonzero vertices. 
By using $n_u\in \mathbb{Z}$, we replace $z_u$ by $h_u n_u$, where 
$h_u>0$ is the vector $\Vec{h}\in \mathbb{R}^I$'s 
component corresponding to vertex $u$.
Let $\Vec{n}\in\mathbb{Z}^I$ be an integer vector.
Then, we have that $F_{\rm MAX}(x)$ is, for all $n_u$,
the sum of the product of $W(x)=\Pr[\bigwedge_{s\in S} W_s \le x]$ and 
probability $\Pr\left[\bigwedge_{u\in I}h_un_u \le W_u\le h_u(n_u+1)\right]$
at the limit $h_u\rightarrow 0$ for all $u\in V$. That is,
we consider the limit $\Vec{h}\rightarrow \Vec{0}$ for 
\begin{align*}
    &\sum_{\Vec{n}\in \mathbb{Z}^I}W(x) \Pr\left[\bigwedge_{u\in I}h_un_u \le W_u\le h_u(n_u+1)\right] \\
    &=\sum_{\Vec{n}\in \mathbb{Z}^I}W(x) \prod_{u\in I}\Pr\left[h_un_u\le W_u \le h_u(n_u+1)\right] \\
    &=\sum_{\Vec{n}\in \mathbb{Z}^I}W(x) \prod_{u\in I}(\Pr[W_u \le h_u(n_u+1)] - \Pr[W_u\le h_un_u]) \\
    &=\sum_{\Vec{n}\in \mathbb{Z}^I}W(x) \prod_{u\in I}h_u\frac{\Pr[W_u \le h_u(n_u+1)] - \Pr[W_u\le h_un_u]}{h_u},
\end{align*}
which tends to $F_{\rm MAX}(x)$.
Remember the definition of the Riemann integration 
\begin{align*}
    \int_\mathbb{R} F(x) dx =\lim_{h\rightarrow 0} \sum_{n\in \mathbb{Z}} h F(nh)
\end{align*}
and the partial derivative with respect to a component $z_u$ of $\Vec{z}$,
\begin{align*}
    \frac{\partial}{\partial z_u}F(\Vec{z})=\lim_{h\rightarrow 0} \frac{F(\Vec{z}+h\mathbf{\rm e}_u)-F(\Vec{z})}{h},
\end{align*}
where $\mathbf{\rm e}_u\in \mathbb{R}^I$ is the orthogonal basis corresponding to $u\in V$. Then, we have
\begin{align*}
    F_{\rm MAX}(x)=\int\limits_{\mathbb{R}^{V\setminus T}} W(x)\prod_{u\in V\setminus T}\frac{\partial}{\partial z_u}\Pr[W_u\le z_u] {\rm d} \Vec{z}.
\end{align*}
Here $\Pr[W_u\le z_u]$ is the distribution function of $W_u$, given by 
\begin{align*}
    \Pr[W_u\le z_u]&=\Pr\left[\max_{v\in V_u}\{X_{uv}+z_v\}\le x\right] =\Pr\left[\bigwedge_{v\in V_u}(X_{uv}+z_v\le x)\right]\\
    &=\Pr\left[\bigwedge_{v\in V_u}(X_{uv}\le z_u-z_v)\right] \\
    &=\prod_{v\in V_u} F_{uv}(z_u-z_v).
\end{align*}
Since we also have 
\begin{align*}
    W(x)\!=\!\Pr\!\left[\bigwedge_{s\in S}\!\!W_s\!\le\! x\right]\!\!=\!\prod_{s\in S}\! \prod_{v\in V_s}\!\! F_{sv}(x\!-\!z_v) 
    \!=\!\!\!\!\int\limits_{\mathbb{R}^S}\!\prod_{u\in S}\!\!H(x\!-\!z_u)\frac{\partial}{\partial z_u}\!\! \prod_{v\in V_u}\!\!F_{uv}(z_u\!-\!z_v) {\rm d}\Vec{z},
\end{align*}
implying the theorem.
\end{proof}

\propositionPsiMax*
\begin{proof}
Notice that $\graphtwo_R$ and $\graphone_R$ shares exactly the same
$s$-side structure, implying that we have $\Psi_U(x,\Vec{y})$
for both of them.
Let us consider a dummy variable $z_u$ for each $u\in U$ 
representing the longest path length from $u$ to the sink.
Also, $y_{uv}$ for $uv\in R$ is the longest path length from
$uv\in R$ to the sink, which is a vertex in both $\graphtwo_R$ and $\graphone_R$.
Then, $\Unif(x-z_u)$ and $H(z_u-y_{uv})$ are the edge length 
distribution functions of $su$ and $(u,uv)$ in both 
$\graphtwo_R$ and $\graphone_R$.

Consider $\Psi_U(x,\Vec{y})$, we set
\begin{align*}
\Psi_U(x,\Vec{y})&=\Pr\left[\bigwedge_{u\in U}\left( X_u\le x-z_u\wedge \bigwedge_{uv\in R}0\le z_u- y_{uv}\right)\right]\\
&=\Pr\left[\bigwedge_{u\in U}\left( X_u\le x-z_u\wedge 0\le z_u-\max_{v\in V_u}\{y_{uv}\}\right)\right].
\end{align*}
By the sifting property of the Dirac Delta, this is transformed into
\begin{align}
&\int_{\mathbb{R}^U}\Pr\left[\bigwedge_{u\in U} X_u\le x-z_u\right]\delta\left(z_u-\max_{v\in V_u}\{y_{uv}\}\right)d\Vec{z}, \nonumber\\
&=\int_{\mathbb{R}^U}\prod_{u\in U}\Unif(x-z_u)\>\>\delta\left(z_u-\max_{v\in V_u}\{y_{uv}\}\right)d\Vec{z}, \label{form:Psi_U1}
\end{align}
since alle edge lengths are mutually independent.
By obserbing $\prod_{v\in V_u}H(z_u-y_{uv})=H(z_u-\max_{v\in V_u}\{y_{uv}\})$,
and $\frac{\partial}{\partial x}H(x)=\delta(x)$,
we have 
\begin{align}
\delta\left(z_u-\max_{v\in V_u}\{y_{uv}\}\right)=\frac{\partial}{\partial z_u} \prod_{v\in V_u}H(z_u-y_{uv}). \label{form:Psi_U2}
\end{align}
The above (\ref{form:Psi_U1}) and (\ref{form:Psi_U2}) leads to (\ref{form:Psi_U0}).

We next consider $\tilde \Psi_V(\Vec{y})$.
Let $z_v$ for each $v\in V$ be the longest path length from
$v$ to the sink of both $\graphtwo_R$ and $\graphone_R$.
We have
\begin{align*}
\tilde \Psi_V(\Vec{y})&=\Pr\left[\bigwedge_{uv\in R} X_{(uv,v)}+X_{vt}\le y_{uv}\right]=\Pr\left[\bigwedge_{uv\in R}\left(X_{(uv,v)}\le y_{uv}-z_v \wedge X_v\le z_v\right)\right].
\end{align*}
Let $\Vec{n}\in \mathbb{Z}^V$ be an integer vector with
components $n_v\in \mathbb{Z}$.
We consider taking the sum of all possibility $z_v\le X_v\le z_v-h$
where $h\rightarrow 0$ and $z_v=hn_v$ for all $\Vec{n}\in \mathbb{Z}^V$, 
which leads to
\begin{align*}
&\lim_{h\rightarrow 0}\sum_{\Vec{n}\in\mathbb{Z}^V}\Pr\left[\bigwedge_{uv\in R}\left(X_{(uv,v)}\le y_{uv}-hn_v \wedge hn_v\le X_v\le hn_v+h\right)\right]\\
&=\lim_{h\rightarrow 0}\sum_{\Vec{n}\in\mathbb{Z}^V}\Pr\left[\bigwedge_{uv\in R}\left(X_{(uv,v)}\le y_{uv}-hn_v \right)\right]\Pr\left[\bigwedge_{v\in V} hn_v\le X_v\le hn_v+h\right],
\end{align*}
by the mutual independence of edge lengths.
Since $\frac{\Pr\left[hn_v\le X_v\le hn_v+h\right]}{h}=\frac{\Unif(hn_v+h)-\Unif(hn_v)}{h}\rightarrow \unif(hn_v)$ when $h\rightarrow 0$ by the definition of the derivative,
the above is equal to 
\begin{align*}
&\lim_{h\rightarrow 0}\sum_{\Vec{n}\in\mathbb{Z}^V}\Pr\left[\bigwedge_{uv\in R}\left(X_{(uv,v)}\le y_{uv}-hn_v\right)\right]\prod_{v\in V}h \unif(n_vh)
\end{align*}
By the definition of the Riemann integral, this is equal to
\begin{align*}
&\int_{\mathbb{R}^V}\Pr\left[\bigwedge_{uv\in R}\left(X_{(uv,v)}\le y_{uv}-z_v\right)\right]\prod_{v\in V}h \unif(z_v)d\Vec{z}.
\end{align*}
By the mutual independence of edge lengths, the above is equal to (\ref{form:tildePsi_V0}).

We proceed to $\Psi_V^+(\Vec{y})$.
We set 
\begin{align*}
\Psi_V^+(\Vec{y})&=\Pr\left[\bigwedge_{uv\in R}X_{vt}+X_t\le y_{uv}\right].
\end{align*}
Observe that 
\begin{align*}
\Pr\left[\bigwedge_{uv\in R}X_{vt}\le y_{uv}-t\right]=\int_{\mathbb{R}^V}\prod_{uv\in R} H(y_{uv}- z_v)\prod_{v\in V} \unif(z_u-t) d\Vec{z},
\end{align*}
by repeathing the above arguments for $\tilde \Psi_V(\Vec{y})$ 
replacing the distribution function $\Unif(y_{uv}-z_v)$ of
the edge length of $(uv,v)$ by $H(y_{uv}-z_v)$ for $uv\in R$, 
and $\unif(z_v)$ by $\unif(z_v-t)$ for $t\in \mathbb{R}$.
Then, we obtain (\ref{form:Psiplus_V0}) by aggregating the 
probability that $X_t\sim\Unif_{|R|}$ is in a small interval 
$[hn,hn+h]$ for $n\in \mathbb{Z}$ and $h\rightarrow 0$.

Finally, by considering the partial derivative of $y_{uv}$ of
$\tilde \Psi_V(\Vec{y})$ for each $uv\in R$, the resulting form implies 
the limit $h\rightarrow 0$ 
of the probability, divided by $h$, 
that the longest path length from $y_{uv}$ is in 
an arbitrarily small interval $[y_{uv},y_{uv}+h]$ for each $uv\in R$.
By aggregating all probabilities multiplied by $\Psi_U(x,\Vec{y})$
for $y_{uv}=n_{uv}h$, where $\Vec{n}\in\mathbb{Z}^{R}$ with components
$n_{uv}\in\mathbb{Z}$ for $uv\in R$,
we have $\tilde \Psi_{\rm MAX}(x)$. 
The similar arguments hold for $\Psi_{\rm MAX}^+(x)$ 
using $\Psi_V^+(\Vec{y})$ instead of $\tilde \Psi_{\rm MAX}(x)$.
\end{proof}

\begin{proposition}
\label{proposition:constant_min} 
Let $G=(V,E)$ be a graph.
We assume that there are $m_0$ edges with static length $0$ and
$m_1=|E|-m_0$ edges with uniform random lengths i.i.d. over $[0,1]$.
If every path in $\Pi_G$ has $k$ nonzero length edges,
we have $\bar F_{\rm MIN}(k-x)=1-F_{\rm MIN}(x)=Ax^{m_1}$ for a constant $A$ 
and $0\le x \le 1$.
\end{proposition}
\begin{proof}
Let $\Vec{X}\in[0,1]^{m1}$ be a random vector with components
$X_i\iidsim \Unif$ for $i=1,\dots,m_1$. 
Since every path $\pi\in \Pi_G$ has $k$ nonzero length edges,
we have
\begin{align*}
    1-F_{\rm MIN}^{\cat{\Vec{0}}{\Vec{X}}}(k-x)=\bar F_{\rm MIN}^{\cat{\Vec{0}}{\Vec{X}}}(k-x)=\bar F_{\rm MIN}^{\cat{\Vec{0}}{\Vec{1}-\Vec{X}}}(-x)=F_{\rm MAX}^{\cat{\Vec{0}}{\Vec{1}-\Vec{X}}}(x),
\end{align*}
by Proposition~\ref{proposition:minmaxsymmetry}.
Since $X_i$ is uniformly distributed over $[0,1]$,
we have $1-X_i\iidsim \Unif$ for $i=1,\dots,i$, 
implying the claim by Proposition~\ref{proposition:constant}. 
\end{proof}






\bibliography{lipics-v2021-sample-article}

\end{document}